\documentclass{article}

\input glyphtounicode
\usepackage[T1]{fontenc}
\usepackage[utf8]{inputenc}

\usepackage{ifdraft}

\usepackage[english]{babel}

\usepackage{enumitem}
\usepackage{multicol}

\usepackage{amsmath}
\usepackage{mathtools}
\usepackage[ntheorem]{empheq}
\usepackage[amsmath,thmmarks,hyperref]{ntheorem}
\usepackage{mleftright}

\usepackage[numbers,sort&compress]{natbib}
\usepackage{xcolor}

\newif\iffancyfont%
\fancyfontfalse%

\IfFileExists{MinionPro.sty}{\IfFileExists{MyriadPro.sty}{\fancyfonttrue}{}}{}

\iffancyfont%
  \usepackage[lf,medfamily,italicgreek,openg]{MinionPro}
  \usepackage[lf,medfamily,onlytext,scale=0.96]{MyriadPro}
  \usepackage[cal=dutchcal]{mathalfa}

  \renewcommand{\dotsm}{{\mathinner{\cdotp\cdotp\cdotp}}}

  \DeclareRobustCommand{\defn}{\mathrel{{\vdotdot}{\equal}}}

  \usepackage[toc,eqno,bib,enum]{tabfigures}
  \newtagform{tabular}[\figureversion{tabular}]{(}{)}
  \usetagform{tabular}
\else%
  \usepackage[charter,greekuppercase=italicized]{mathdesign}

  \DeclareFontFamily{U}{dutchcal}{\skewchar \font =45}
  \DeclareFontShape{U}{dutchcal}{m}{n}{<-> s*[1.0] dutchcal-r}{}
  \DeclareFontShape{U}{dutchcal}{b}{n}{<-> s*[1.0] dutchcal-b}{}
  \DeclareSymbolFont{dutchcal}{U}{dutchcal}{m}{n}
  \SetSymbolFont{dutchcal}{bold}{U}{dutchcal}{b}{n}
  \DeclareSymbolFontAlphabet{\mathcal}{dutchcal}

  \newcommand\tbfigures\null

  \let\uppi\piup
  \let\upDelta\Deltaup

  \newcommand{\defn}{\coloneqq}
  
\fi%

\usepackage{bm}

\usepackage{dsfont}

\usepackage[final]{microtype}

\allowdisplaybreaks

\usepackage[sf,bf,small]{titlesec}
\usepackage[labelfont={sf,bf},labelsep=period]{caption}

\setlist[1]{labelindent=\parindent}
\setlist[description]{font=\sffamily\bfseries,align=right,labelsep=1em}

\numberwithin{equation}{section}

\setlist[enumerate,1]{label=(\roman*),font=\normalfont}

\makeatletter
\newcounter{and}
\newdimen{\instindent}
\newcommand{\institute}[1]{\newcommand{\@institute}{#1}}
\newcommand{\inst}[1]{\unskip\smash{$^{#1}$}\setcounter{and}{1}\ignorespaces}
\newcommand{\email}[1]{\href{mailto:#1}{#1}}
\renewcommand{\maketitle}{
  {
    \raggedright%
    \LARGE%
    \noindent%
    \bfseries%
    \sffamily%
    \@title%
    \par
  }

  \vspace{1.5\baselineskip}

  {
    \raggedright%
    \renewcommand{\and}{\unskip, \ignorespaces}%
    \noindent\ignorespaces\@author\par
  }

  \vspace{0.5\baselineskip}

  {
    \small%
    \parindent=0pt%
    \parskip=0pt%
    \setcounter{and}{1}%
    \renewcommand{\and}{%
      \par\stepcounter{and}%
      \hangindent\instindent%
      \noindent%
      \hbox to \instindent{\hss\smash{$^{\theand}$\enspace}}\ignorespaces%
    }%
    \setbox0=\vbox{\@institute}%
    \ifnum\value{and}>9\relax\setbox0=\hbox{$^{88}$\enspace}%
    \else\setbox0=\hbox{$^{8}$\enspace}\fi%
    \instindent=\wd0\relax%
    \ifnum\value{and}=1\relax%
    \else%
      \setcounter{and}{1}%
      \hangindent\instindent%
      \noindent%
      \hbox to \instindent{\hss\smash{$^{\theand}$}\enspace}\ignorespaces%
    \fi%
    \ignorespaces%
    \@institute\par
  }
}
\makeatother

\renewenvironment{abstract}{
  \addvspace{1.5\baselineskip}%
  \topsep=0pt\partopsep=0pt%
  \trivlist\item[\hspace{\labelsep}\bfseries\sffamily Abstract.]
}{}

\usepackage[obeyDraft]{todonotes}

\ifdraft{
  \usepackage[notref,notcite]{showkeys}

  \usepackage{filemod}
  \usepackage{fancyhdr}
  \pagestyle{fancy}
  \fancyhf{}
  
  \fancyhead[C]{Draft version. Last modified: \filemodprint{\jobname}}
  \fancyfoot[C]{\thepage}
}{}

\newcommand{\e}{\mathrm{e}}
\newcommand{\im}{\mathrm{i}}

\newcommand{\field}[1][K]{{\mathds{#1}}}
\newcommand{\NN}{{\field[N]}}

\newcommand{\RR}{{\field[R]}}
\newcommand{\CC}{{\field[C]}}

\renewcommand{\Re}{\operatorname{Re}}
\renewcommand{\Im}{\operatorname{Im}}

\newcommand{\dif}{\mathop{}\!\mathrm{d}}

\DeclareMathOperator{\supp}{supp}

\DeclareMathOperator{\sgn}{sgn}

\DeclareMathOperator{\csch}{csch}

\newcommand{\conj}[1]{\overline{#1}}
\DeclarePairedDelimiter{\abs}{\lvert}{\rvert}
\DeclarePairedDelimiter{\norm}{\lVert}{\rVert}

\DeclarePairedDelimiter{\floor}{\lfloor}{\rfloor}
\DeclarePairedDelimiter{\ceil}{\lceil}{\rceil}

\newcommand{\init}{\mathrm{init}}
\newcommand{\fin}{\mathrm{fin}}
\newcommand{\free}{\mathrm{free}}
\newcommand{\tow}{\mathrm{tow}}

\newcommand{\one}{\mathds{1}}

\let\vec\bm

\newcommand{\dotmatrix}[4]{%
  \sbox0{$+$}%
  \sbox2{$\cdot$}%
  \mathbin{\ooalign{%
    \phantom{$+$}\cr%
    \ifx1#1\hidewidth\raise\dimexpr+\ht0/2\relax\hbox{$\cdot$}\hidewidth\cr\fi%
    \ifx1#2\hidewidth\kern\dimexpr-\wd0+\wd2/2\relax\hbox{$\cdot$}\hidewidth\cr\fi%
    \ifx1#3\hidewidth\kern\dimexpr\wd0-\wd2/2\relax\hbox{$\cdot$}\hidewidth\cr\fi%
    \ifx1#4\hidewidth\raise\dimexpr-\ht0/2\relax\hbox{$\cdot$}\hidewidth\cr\fi%
  }}%
}

\makeatletter
\newtheoremstyle{nonumberplainnoparens}%
  {\item[\theorem@headerfont\hskip\labelsep ##1\theorem@separator]}%
  {\item[\theorem@headerfont\hskip\labelsep ##1 ##3\theorem@separator]}
\makeatother

\theoremnumbering{arabic}
\qedsymbol{\ensuremath{\quad\Box}}

\theoremstyle{plain}

\theorembodyfont{\itshape}
\theoremheaderfont{\normalfont\sffamily\bfseries}
\theoremseparator{.}

\newtheorem{theorem}{Theorem}[section]
\renewtheorem*{theorem*}{Theorem}
\newtheorem{proposition}[theorem]{Proposition}
\newtheorem{lemma}[theorem]{Lemma}

\theorembodyfont{\normalfont}

\newtheorem{remark}[theorem]{Remark}

\theoremstyle{nonumberplainnoparens}
\theoremsymbol{\ensuremath{\quad\Box}}
\newtheorem{proof}{Proof}

\definecolor{hypercolor}{rgb}{0,0.2,0.7}
\usepackage[breaklinks,final]{hyperref}

\hypersetup{
  unicode,
  colorlinks,
  linkcolor=hypercolor,
  urlcolor=hypercolor,
  citecolor=hypercolor,
  bookmarksnumbered,
  bookmarksdepth=2,
  pdfauthor={Hanno Gottschalk and Daniel Siemssen},
  pdftitle={The Cosmological Semiclassical Einstein Equation as an Infinite-Dimensional Dynamical System}
}

\title{The Cosmological Semiclassical Einstein Equation as an Infinite-Dimensional Dynamical System}

\author{
  Hanno Gottschalk\inst{1}
  \and
  Daniel Siemssen\inst{2}
}

\institute{
  Department of Mathematics and Informatics, University of Wuppertal, Gau{\ss}stra{\ss}e 20, 42119 Wuppertal, Germany.
  E-mail:~\email{hanno.gottschalk@uni-wuppertal.de}.
  \and
  Department of Mathematics, University of York, Heslington, York YO10 5DD, United Kingdom.
  E-mail:~\email{daniel.siemssen@york.ac.uk}.
}

\begin{document}

\maketitle

\begin{abstract}
  We develop a comprehensive framework in which the existence of solutions to the semiclassical Einstein equation (SCE) in cosmological spacetimes is shown.
  Different from previous work on this subject, we do not restrict to the conformally coupled scalar field and we admit the full renormalization freedom.
  Based on a regularization procedure, which utilizes homogeneous distributions and is equivalent to Hadamard point-splitting, we obtain a reformulation of the evolution of the quantum state as an infinite-dimensional dynamical system with mathematical features that are distinct from the standard theory of infinite-dimensional dynamical systems (e.g., unbounded evolution operators).
  Nevertheless, applying methods closely related to Ovsyannikov's method, we show existence of maximal/global solutions to the SCE for vacuum-like states, and of local solutions for thermal-like states.
  Our equations do not show the instability of the Minkowski solution described by other authors.
\end{abstract}

\section{Introduction}

The \emph{semiclassical Einstein equation (SCE)} is the equation
\begin{equation}\label{eq:SCE}
  G_{\mu\nu} = \kappa {\langle T^\mathrm{ren}_{\mu\nu} \rangle}_\omega,
\end{equation}
where $G_{\mu\nu}$ denotes the Einstein tensor for the (classical) metric $g_{\mu\nu}$, $\kappa$ is a gravitational coupling constant constant (conventionally, $\kappa = 8\uppi\mathrm{G}$), and $\langle T^\mathrm{ren}_{\mu\nu} \rangle_\omega$ is the renormalized stress-energy tensor for a quantum field theory (QFT) in the state $\omega$.
That is, matter is described by a quantum field and gravity is described by a classical Lorentzian manifold.
The SCE has been studied since the early 1960's by a number of authors, see~\cite{hack:book,wald:qft-book,birrell-davies} for an overview.
It is typically introduced in an ad hoc manner as a minimal change of the classical Einstein equation by replacing the stress-energy tensor of a classical field by that of a quantum field to take into account the quantum nature of matter.
In particular, it is not considered a fundamental equation but rather an approximation of a more fundamental theory within some domain of validity that is sufficiently remote from the Planck scale.
Some possible derivations of the SCE from a quantum gravity are critically discussed in Sect.~II.B of~\cite{flanagan-wald}.

Many equations in QFT are plagued by ultraviolet divergences and the SCE is no different, because (na\"ively) the expectation value of the stress-energy tensor involves the evaluation of singular quantum fields at a point.
As already discussed in~\cite{wald:back-reaction}, a renormalization of the stress-energy tensor needs to be coordinate independent and thus has to follow the principle of general covariance.

A procedure which satisfies these requirements is the point-splitting formalism by Christensen~\cite{christensen:1,christensen:2}.
Here one subtracts the singular part, given by the Hadamard parametrix (essentially a Lorentzian version of the heat kernel), from the two-point function of the state.
For this reason one restricts the class of states to Hadamard states, viz., states with two-point functions that match the Hadamard parametrix up to smooth contributions.
(Actually it is sufficient that the leading singularities of the two-point function match those of the Hadamard parametrix, as in the case of adiabatic states~\cite{junker-schrohe}.)

Due to the ambiguity in the regularization procedure (satisfying certain conditions or axioms~\cite{hollands-wald:wick,hollands-wald:time-order,hollands-wald:stress-energy}), a renormalization freedom arises.
Some terms renormalize the gravitational constant or the cosmological constant by a finite amount.
Thus it can be seen that the wide-spread belief that quantum matter automatically leads to a very large cosmological constant is not correct; instead the cosmological constant corresponds to a renormalization freedom.
Other terms contain higher than second order derivatives in the metric.
Such terms are likely to change the entire characteristic of the SCE, especially with respect to the classical Einstein equation, which is only of second order, and there seems to be to be no justifiable reason why these higher order terms should be discarded~\cite{wald:trace}.
We would like to mention, however, the method of order reduction by which higher order derivatives in the SCE can be replaced by lower order derivatives in a systematic way, see~\cite{flanagan-wald} and references therein.

Moreover, it was noted early on, see e.g.~\cite{christensen:1,wald:trace}, that the renormalized stress-energy tensor is not traceless, even if the classical action of the quantum field is conformally invariant.
In fact, in this case one obtains the famous \emph{trace (or conformal) anomaly}
\begin{equation*}
  \langle T^\mathrm{ren} \rangle = g^{\mu\nu} \langle T_{\mu\nu}^\mathrm{ren} \rangle = c_1 R^2 + c_2 R_{\mu\nu} R^{\mu\nu} + c_3 R_{\mu\nu\lambda\sigma} R^{\mu\nu\lambda\sigma} + c_4 \Box R + \text{renormalization freedom},
\end{equation*}
where the constants~$c_\bullet$ depend on spin of the (free) quantum field.
Notably, the right-hand side does not depend on the quantum state at all, and contains higher than second order derivatives in the metric.

In the following we shall specialize our discussion to the special case of a free scalar field on flat \emph{Friedmann--Lemaître--Robertson--Walker (FLRW)} spacetimes $M = I \times \RR^3$, $I \subset \RR$, with the metric (in conformal time~$\tau$ and with signature convention ${-}{+}{+}{+}$)
\begin{equation}\label{eq:flrw}
  g = a(\tau)^2 \bigl(-\dif\tau^2 + \dif\vec{x}^2\bigr),
\end{equation}
where $a(\tau) > 0$ is called the scale factor and $\dif\vec{x}^2$ denotes the Euclidean metric on~$\RR^3$.
This case already shows many features that distinguish it from QFT on the maximally symmetric Minkowski and de Sitter spacetime, or static spacetimes.
Additionally, these spacetimes are of importance in cosmology as they represent the observed homogeneous and isotropic structure of our universe at the scale of several megaparsec, as well as the observed flatness \cite{weinberg}.

Despite some mathematical problems in the pre-1990's literature on the SCE, by the beginning of the 1980's the approach has been developed to a stage where numerical solutions and cosmological applications of the SCE were in reach.
This situation was exploited by Anderson in a series of four papers~\cite{anderson:1,anderson:2,anderson:3,anderson:4} starting with the conformally coupled and massless scalar field following up prior work by Starobinski~\cite{starobinsky}.
Depending on the values of the aforementioned renormalization parameters, Anderson discovered a very rich behavior of the SCE ranging from big bang, big bounce to divergence of the scale factor to infinity in finite time.
The non-conformally coupled case was investigated by analytical and numerical methods by Suen~\cite{suen:stability1,suen:stability2}.
In his work he discusses an instability of Minkowski spacetime as a global solution to the SCE (in the sense of continuous dependence of the solution on its initial data).
Further analytical and numerical work on the stability of the SCE has been performed by H{\"a}nsel and Verch~\cite{hansel-phd}.

As shown by Fulling, Sweeny and Wald in~\cite{fulling-sweeny-wald} (with some improved results by Radzikowski~\cite{radzikowski,radzikowski-verch} and others), a state which is Hadamard on a time-slice around a Cauchy surface is Hadamard on the entire globally hyperbolic manifold containing this surface.
From the perspective of the SCE, this result gives an important hint for the well-posedness of the former equation ensuring that the stress-energy tensor will remain well-defined on the entire spacetime manifold.

The proper description of the stress-energy tensor in the mathematically rigorous framework paved the way for new cosmological investigations on the SCE.
Dappiaggi, Fredenhagen and Pinamonti~\cite{dappiaggi-fredenhagen-pinamonti} investigated the stability of the SCE on FLRW spacetimes using certain effective large-mass approximations of the quantum state.
As a major milestone in the mathematical theory of the SCE, Pinamonti proved the existence of solutions for the trace equation for short times in the conformally coupled case for certain states defined on past null infinities~\cite{pinamonti}.
This has been significantly extended by Pinamonti and one of us~\cite{pinamonti-siemssen} to full solutions of the SCE on FLRW spacetimes, including the energy constraint with initial values specified on a Cauchy surface, but still limited to the conformally coupled scalar field, and a particular choice of renormalization parameters and quantum state.
Other recent work related to the SCE include: \cite{nicolai}, where numerical solutions to the equations presented here are analyzed; \cite{paolo}, where yet another approach to solving the SCE is developed; \cite{juarez=aubry-miramontes-sudarsky}, where a simplified semiclassical problem is studied as an initial value problem; \cite{sanders} and \cite{juarez=aubry:static}, where solutions in static spacetimes are analyzed.

We remark that in this paper we deal with the SCE without a classical Klein--Gordon field.
However, as such a field can be important in inflationary cosmology (see e.g.~\cite{hack:book,mukhanov:book,weinberg}, we give remarks on the minor changes that are needed to include it without changing the results of our paper.

\subsection{Outline}

After this introduction, in the second section we first outline the quantization of the (real, free) scalar field in curved spacetimes with emphasis on flat FLRW spacetimes.
In particular, we present an initial value formulation for homogeneous isotropic states for the quantum scalar field in FLRW spacetimes \cite{luders-roberts}.
Then, we develop a point-splitting regularization specially adapted to FLRW spacetimes and show its equivalence to the conventional Hadamard point-splitting.

In the third section, we briefly discuss the renormalized stress-energy tensor for the quantized scalar field.
Since the trace and the energy component of the stress-energy tensor play an important role for the (semiclassical) Einstein equation in FLRW spacetimes, we present the corresponding expressions.
We also discuss the problem posed by higher derivatives present in the semiclassical theory, partially due to regularization and renormalization.

Next, in the fourth section, we show how to formulate and solve a dynamical system for the coincidence limit of the regularized two-point function and its derivatives.
This system shows various interesting physical and mathematical features:
First, it hides the higher derivatives present in the regularization, thereby circumventing one of the challenges faced when solving the SCE.
Next, it distinguishes a class of `generalized' vacuum states from other classes of states such as thermal states.
Finally, its evolution is given by a (generally unbounded) evolution operator which can be understood as acting between differently weighted sequence spaces.
As was pointed out to us after the completion of this work, the construction of the evolution operator can be seen as an application of Ovsyannikov's method \cite{ovsyannikov,friesen} of solving Cauchy problems in scales of Banach spaces.\footnote{Ovsyannikov's method is essentially an abstract generalization of the Cauchy--Kovalevskaya theorem. In applications to hyperbolic equations, for example, it makes it possible to treat solutions which are analytic in spatial directions but (different than in the Cauchy--Kovalevskaya approach) only continuously differentiable in the time direction.}
The fact that this evolution operator exists at all is intimately tied to the hyperbolicity of the Klein--Gordon operator, which expresses itself in our dynamical system through the nilpotence of a certain matrix.
Depending on the choice of the weights in the sequence spaces, the evolution can be shown to exist for all time (geometrically growing weights) or for a finite amount of time (factorially growing weights).

Finally, in the fifth section, we use the just developed dynamical system for the quantum state to solve the SCE.
In fact, we consider an abstract class of quasilinear equations which includes as a special case the SCE.
For this abstract equation we show existence and uniqueness of solutions, as well as continuous dependence on the initial values and parameters of the equation.
Existence is shown in finite time intervals with a priori bounds depending on the choice of initial values, in particular the initial values for the quantum state.
We analyze this equation for various possible choices of parameters in the case of the SCE in FLRW spacetimes.
In general, the resulting equation is of fourth order but in the special case of conformal coupling it can reduce to a second order equation.
We also remark on the instability of Minkowski spacetime stating that arbitrarily small perturbations in the initial conditions can lead to finite effects in the solution of the energy equation \cite{suen:stability1,suen:stability2} -- no such instability appears in our approach based on the trace equation, which we show to depend continuously on the initial data.
At the end, we explain a method of constructing physical initial data and give a short outlook on future research topics.
Let us state our main results, summarizing the contents of Thms.~\ref{thm:traced_sce} and~\ref{thm:initial_data}:
\begin{theorem*}
  There exists a non-empty set of self-consistent initial data for the SCE for flat FLRW spacetimes, and, given such self-consistent initial data, the SCE has locally a unique smooth solution which depends continuously on the initial data and parameters.
  Furthermore, the quantum state part of the solution is of Hadamard type.
  For a class of `vacuum-like' initial data for the state, the solution is maximal or even global.
\end{theorem*}

In the appendix we list several auxiliary results (concerning homogeneous distributions, weighted sequence spaces, combinatorial inequalities and Synge's world function) used throughout this work.

\section{Scalar field on flat FLRW}
\label{sec:scalar}

\subsection{Klein--Gordon equation}

Consider the (homogeneous) Klein--Gordon equation
\begin{equation}\label{eq:klein-gordon}
  K \phi \defn (\Box + \xi R + m^2) \phi = 0
\end{equation}
with mass~$m \geq 0$ and curvature coupling~$\xi$ ($\xi=0$ is called \emph{minimal coupling} and $\xi=\frac16$ \emph{conformal coupling}).
We define the d'Alembert operator as $\Box \defn -g^{\mu\nu} \nabla_{\!\mu} \nabla_{\!\nu}$, where $\nabla$ is the covariant derivative and we employ Einstein summation convention.

In conformal time~$\tau$, the \emph{flat} FLRW metric takes the form~\eqref{eq:flrw}.
The d'Alembertian for this metric is $\Box = a^{-3} (\partial_\tau^2 - a^{-1} a'' - \upDelta) a$,
where we denote derivatives with respect to conformal time by primes and $\upDelta$ is the Laplace operator on $\RR^3$.
Introducing the conformally rescaled field $\varphi$ and the potential $V$,
\begin{equation*}
  \varphi \defn a \phi,
  \quad
  V \defn (6\xi-1) \frac{a''}{a} + a^2 m^2,
\end{equation*}
the Klein--Gordon equation~\eqref{eq:klein-gordon} can thus be written as $(\partial_\tau^2 - \upDelta + V) \varphi = 0$, where we used $R = 6 a^{-3} a''$.
Further rewritten in Hamiltonian form, this equation becomes
\begin{equation}\label{eq:1st-order-KG}
  \partial_\tau \begin{pmatrix}
    \varphi \\ \pi
  \end{pmatrix}
  =
  \begin{pmatrix}
    0 & 1 \\ \upDelta - V & 0
  \end{pmatrix}
  \begin{pmatrix}
    \varphi \\ \pi
  \end{pmatrix},
  \quad
  \pi \defn \varphi'.
\end{equation}

\subsection{Quantization}

For the quantization of the (real, free) scalar field $\phi$ on a globally hyperbolic spacetime $(M,g)$, we follow the algebraic approach, see e.g.\ \cite{dimock,hack:book,khavkine-moretti}:
We generate a (non-commutative, unital) $*$-algebra~$\mathcal{A}$ by `smeared' \emph{quantum fields} $\hat\phi(f)$ for $f \in C^\infty_{\mathrm c}(M)$ satisfying
\begin{enumerate}
  \item linearity: $\hat\phi(\alpha f + \beta g) = \alpha \hat\phi(f) + \beta \hat\phi(g)$,
  \item hermiticity: $\hat\phi(f)^* = \hat\phi(\conj{f})$,
  \item Klein--Gordon equation: $\hat\phi(Kf) = 0$,
  \item canonical commutation relations (CCR): $[\hat\phi(f),\hat\phi(g)] = -\im \langle f \,|\, G^\mathrm{PJ} g \rangle \one$,
\end{enumerate}
where $f,g \in C^\infty_{\mathrm c}(M)$ and $\alpha,\beta \in \CC$.
Moreover, $\langle \cdot\,|\,\cdot \rangle$ is the canonical $L^2$ product on the spacetime, and $G^\mathrm{PJ}$ denotes the uniquely defined Pauli--Jordan propagator~\cite{bar-ginoux-pfaffle,derezinski-siemssen}, namely the difference of forward (retarded) and backward (advanced) propagator of the Klein--Gordon operator~$K$.

A \emph{state} on $\mathcal{A}$ is a linear functional $\omega : \mathcal{A} \to \CC$, which is
\begin{enumerate}
  \item normalized [$\omega(\one) = 1$] and
  \item positive [$\omega(a^* a) \geq 0$ for all $a \in \mathcal{A}$],
\end{enumerate}
The two-point function $\omega_2$ of $\omega$ is defined as $\omega_2(f, g) \defn \omega(\hat\phi(f)\hat\phi(g))$.
Due to the positivity of the state we have $\omega_2(\conj{f}, f) \geq 0$.
Meanwhile, the canonical commutation relations imply
\begin{equation*}
  \omega_2(f,g)-\omega_2(g,f) = -\im \langle f \,|\, G^\mathrm{PJ} g \rangle.
\end{equation*}

On Minkowski spacetime, and more generally on static spacetimes, one additionally imposes a spectrum condition (positive frequency condition) for the state, which distinguishes a vacuum state.
On generic spacetimes, no `natural' preferred state exists~\cite{fewster-verch} and (for free fields) the spectrum condition is replaced by a condition on the singular structure of the two-point function.
Typically one requires that the state is a \emph{Hadamard state}, viz., its two-point function satisfies the \emph{microlocal spectrum condition} -- a condition on the smooth wave-front set~\cite{radzikowski}.
In applications it is sometimes useful to relax the microlocal spectrum condition.
For example, on FLRW spacetimes one often considers the class of \emph{adiabatic states} which are obtained via a WKB-type approach~\cite{parker:1,junker-schrohe}.
These states satisfy a Sobolev version of the microlocal spectrum condition~\cite{junker-schrohe}.

\begin{remark}\label{rem:background-fields}
  Here and in the following we assume that the state $\omega$ has a vanishing one-point function $\omega(\hat\phi(f))=0$.
  Thus we do not distinguish between the two-point function and the connected two-point function.
  A non-vanishing one-point function $\phi^\mathrm{bg}(\tau,\vec{x}) \defn \omega(\hat\phi(\tau,\vec{x}))$ is interpreted as a classical Klein--Gordon `background' field.
  In the context of homogeneous and isotropic spacetimes, $\phi^\mathrm{bg}(\tau)$ does not depend on $\vec{x}$.
  Setting $\varphi^\mathrm{bg} \defn a\phi^\mathrm{bg}$ and $\pi^\mathrm{bg} = \partial_\tau \varphi^\mathrm{bg}$, we see that the dynamics of the additional two degrees of freedom introduced by the classical Klein--Gordon field is given by~\eqref{eq:1st-order-KG}, where the spatial Laplacian $\upDelta$ can be omitted.
\end{remark}

\subsection{Two-point functions}

Consider the two-point function $\omega_2$ of a homogeneous and isotropic state.
That is, it holds that
\begin{equation}\label{eq:omega2-homogeneous}
  \omega_2\bigl((\tau,\vec{x}),(\eta,\vec{y})\bigr) = \omega_2\bigl(\tau, \eta, r = \abs{\vec{x}-\vec{y}}\bigr).
\end{equation}

Since the antisymmetric part of the two-point function is fixed by the commutator `function' (viz., the Pauli--Jordan propagator $G^\mathrm{PJ}$), $\omega_2$ is completely determined by its symmetric part
\begin{equation}\label{eq:omega2-symm}
  \frac12 \bigl( \omega_2(\tau, \eta, r) + \omega_2(\eta, \tau, r) \bigr).
\end{equation}
Therefore, the Cauchy data at conformal time~$\tau$ of a homogeneous and isotropic state can be given by~\eqref{eq:omega2-symm} and its first time derivatives.

We define
\begin{equation}\label{eq:G-data}
  \mathcal{G}(\tau,r)
  \defn
  \begin{pmatrix}
    \mathcal{G}_{\varphi\varphi}(\tau,r) \\
    \mathcal{G}_{(\varphi\pi)}(\tau,r) \\
    \mathcal{G}_{\pi\pi}(\tau,r)
  \end{pmatrix}
  \defn
  \lim_{\eta \to \tau}
  \begin{pmatrix}
    \one \\
    \frac12 (\partial_\tau + \partial_\eta) \\
    \partial_\tau \partial_{\eta}
  \end{pmatrix}
  a(\tau) a(\eta) \omega_2(\tau,\eta,r),
\end{equation}
which represents the Cauchy data of the two-point function at~$\tau$.
Using Synge's rule (to pull derivatives into the coincidence limit for the time variables) and the Hamiltonian form of the Klein--Gordon equation~\eqref{eq:1st-order-KG}, one can then rewrite the equations of motion for the two-point function as
\begin{equation}\label{eq:G-dynamics}
  \partial_\tau \mathcal{G} = \begin{pmatrix}
    0 & 2 & 0 \\ \upDelta_r - V & 0 & 1 \\ 0 & 2(\upDelta_r - V)  & 0
  \end{pmatrix} \mathcal{G},
\end{equation}
where $\upDelta_r \defn r^{-2} \partial_r r^2 \partial_r$ is the (three dimensional) radial Laplacian.

It is sometimes convenient to perform calculations in momentum space.
We define the mode functions
\begin{equation}\label{eq:G-Fourier}
  \widehat{\mathcal{G}}(\tau,k)
  \defn
  \begin{pmatrix}
    \widehat{\mathcal{G}}_{\varphi\varphi}(\tau,k) \\
    \widehat{\mathcal{G}}_{(\varphi\pi)}(\tau,k) \\
    \widehat{\mathcal{G}}_{\pi\pi}(\tau,k)
  \end{pmatrix}
  \defn 4\uppi \int_0^\infty \mathcal{G}(\tau,r) \frac{\sin(k r)}{k r} r^2 \dif r
\end{equation}
with $k \in [0,\infty)$.
The mode functions are simply the Fourier transform of $\mathcal{G}(\tau,r=\abs{\vec{x}})$ in $\vec{x}$.
Indeed, using the convention $\hat{f}(\vec{k}) \defn \int_{\RR^3} f(\vec{x}) \e^{-\im\vec{k}\cdot\vec{x}} \dif\vec{x}$ for the Fourier transform, we have for radial functions (or distributions) $f(\vec{x}) = f(r = \abs{\vec{x}})$
\begin{equation*}
  \hat{f}(\vec{k}) = \hat{f}(k = \abs{\vec{k}}) = 4\uppi \int_0^\infty f(r) \frac{\sin(kr)}{kr} r^2 \dif r.
\end{equation*}

Conversely, the mode functions specify the two-point function by
\begin{equation}\label{eq:G-Fourier-inv}
  \mathcal{G}(\tau,r) = \frac{1}{2\uppi^2} \int_0^\infty \widehat{\mathcal{G}}(\tau,k) \frac{\sin(k r)}{k r} k^2 \dif k.
\end{equation}
It follows immediately from~\eqref{eq:G-dynamics} that the mode functions solve the dynamical equations
\begin{equation}\label{eq:G-data-fourier}
  \partial_\tau \widehat{\mathcal{G}} = \begin{pmatrix}
    0 & 2 & 0 \\ -(k^2 + V) & 0 & 1 \\ 0 & -2(k^2 + V)  & 0
  \end{pmatrix} \widehat{\mathcal{G}}.
\end{equation}

A straightforward computation shows that $\widehat{\mathcal{J}} \defn \widehat{\mathcal{G}}_{\varphi\varphi} \widehat{\mathcal{G}}_{\pi\pi} - \widehat{\mathcal{G}}_{(\varphi\pi)}^{\;2}$ is a conserved quantity.
This reduces the degrees of freedom in~\eqref{eq:G-data} (and~\eqref{eq:G-data-fourier}) to two.
Moreover, it follows from the positivity of the state that $\widehat{\mathcal{J}} \geq \frac14$ with equality for pure states~\cite{luders-roberts}.

\subsection{A point-splitting regularization}

Many expressions of physical relevance in QFT involve expectation values of products of quantum fields at a point.
Na\"ively, such expressions are ill-defined because of the distributional nature of quantum fields.
By restricting to a class of states which share a common singular structure, we can define a renormalization scheme that allows to make sense of these expressions.

Below we develop a regularization scheme for two-point functions on FLRW spacetimes (or, in fact, for the kernels $\mathcal{G}(\tau,r)$) which carries features of both the \emph{Hadamard point-splitting} method~\cite{christensen:1,christensen:2} and the WKB-type approach named \emph{adiabatic regularization} \cite{parker:1,anderson-parker,bunch:adiabatic}.
More concretely, the aim of this subsection is to define kernels $\mathcal{H}_n(\tau,r)$ such that $\mathcal{G}(\tau,r) - \mathcal{H}_n(\tau,r)$ is sufficiently regular in the limit $r \to 0$.

Let $\mu$ be an arbitrary (but fixed) length scale.
On~$\RR$, define the piecewise function
\begin{equation*}
  k^z_+ \defn \begin{dcases*}
    k^z & if $k > 0$, \\
    0 & if $k \leq 0$.
  \end{dcases*}
\end{equation*}
For $z \in \CC \setminus \{-1,-2,\dotsc\}$ this defines a distribution (by analytic continuation, cf.\ Chap. III.2 of~\cite{hormander:1}).
It can be extended to all $z \in \CC$ by defining for $n \in \NN$,
\begin{equation*}
  \langle k_+^{-n}, f \rangle \defn \frac{1}{(n-1)!} \biggl( -\int_0^\infty \log(\mu k) f^{(n)}(k) \dif k + f^{(n-1)}(0) \sum_{j=1}^{n-1} \frac{1}{j} \biggr),
  \quad
  f \in C^\infty_{\mathrm c}(\RR).
\end{equation*}
We also define for $r \geq 0$ and $z \in \CC$ the distributions
\begin{equation*}
  h_z(r) \defn \frac{\e^{\im z \uppi/2}}{2\uppi^2} \frac{r^{z-2}}{\Gamma(z)} \biggl( \log\Bigl(\frac{r}{\mu}\Bigr) - \psi(z) \biggr),
\end{equation*}
where $\psi$ denotes the Digamma function.
Note that $h_{-2} = -{(\uppi^2r^4)}^{-1}$ and $h_{0} = 1/{(2\uppi^2r^2)}^{-1}$.
Further details about these functions can be found in the appendix (Sect.~\ref{sec:homogeneous}), although without the constant length scale~$\mu$.
Tacitly, the so defined distributions and the resulting regularization scheme depend on $\mu$.

We make the Ansatz (equivalently either in position or momentum space with relations analogous to~\eqref{eq:G-Fourier} and~\eqref{eq:G-Fourier-inv})
\begin{subequations}\label{eq:H-expansion}\begin{align}
  \mathcal{H}_n(\tau,r) &\defn
  \begin{pmatrix*}
    \mathcal{H}_{\varphi\varphi,n}(\tau,r) \\
    \mathcal{H}_{(\varphi\pi),n}(\tau,r) \\
    \mathcal{H}_{\pi\pi,n}(\tau,r)
  \end{pmatrix*}
  =
  \begin{pmatrix}
    0 \\ 0 \\ \gamma_{-1}(\tau)
  \end{pmatrix}
  h_{-2}(r)
  +
  \sum_{j=0}^n
  \begin{pmatrix}
    \alpha_j(\tau) \\ \beta_j(\tau) \\ \gamma_j(\tau)
  \end{pmatrix}
  h_{2j}(r), \\
  \widehat{\mathcal{H}}_n(\tau,k) &\defn
  \begin{pmatrix*}
    \widehat{\mathcal{H}}_{\varphi\varphi,n}(\tau,k) \\
    \widehat{\mathcal{H}}_{(\varphi\pi),n}(\tau,k) \\
    \widehat{\mathcal{H}}_{\pi\pi,n}(\tau,k)
  \end{pmatrix*}
  =
  \begin{pmatrix}
    0 \\ 0 \\ \gamma_{-1}(\tau)
  \end{pmatrix}
  k_+^1
  +
  \sum_{j=0}^n
  \begin{pmatrix}
    \alpha_j(\tau) \\ \beta_j(\tau) \\ \gamma_j(\tau)
  \end{pmatrix}
  k_+^{-(2j+1)},
\end{align}\end{subequations}
where we fix the lowest order parameters to
\begin{equation}\label{eq:H-init-coeffs}
  \alpha_0 = \frac12,
  \quad
  \beta_0 = 0,
  \quad
  \gamma_{-1} = \frac12.
\end{equation}

We wish to find coefficients $\alpha_\bullet$, $\beta_\bullet$ and $\gamma_\bullet$ such that the two constraints
\begin{align}
  \label{eq:H-dynamics}
  \partial_\tau \mathcal{H}_n
  -
  \begin{pmatrix}
    0 & 2 & 0 \\ \upDelta_r - V & 0 & 1 \\ 0 & 2(\upDelta_r - V)  & 0
  \end{pmatrix}
  \mathcal{H}_n
  &=
  \mathcal{O}(r^{2(n-1)})
  \\
  \label{eq:H-pure}
  \widehat{\mathcal{H}}_{\varphi\varphi,n} \widehat{\mathcal{H}}_{\pi\pi,n} - \widehat{\mathcal{H}}_{(\varphi\pi),n}^{\,2} &= \frac{1}{4} + \mathcal{O}(k^{-2(n+1)})
\end{align}
are satisfied.
That is, $\mathcal{H}_n$ satisfies the Klein--Gordon equation and has the properties of the two-point function of a pure state up to any desired order.

Using the homogeneity property~\eqref{eq:homog_h} and~\eqref{eq:H-init-coeffs}, we find
\begin{align*}\MoveEqLeft
  \partial_\tau \mathcal{H}_n
  -
  \begin{pmatrix}
    0 & 2 & 0 \\ \upDelta_r - V & 0 & 1 \\ 0 & 2(\upDelta_r - V)  & 0
  \end{pmatrix}
  \mathcal{H}_n \\
  &=
  \sum_{j=0}^{n-1}
  \begin{pmatrix}
    \alpha_j' - 2 \beta_j \\
    \beta_j' + \alpha_{j+1} + V \alpha_j - \gamma_j \\
    \gamma_j' + 2 \beta_{j+1} + 2 V \beta_j
  \end{pmatrix}
  h_{2j}
  +
  \begin{pmatrix}
    \alpha_n' - 2 \beta_n, \\
    \beta_n' + V \alpha_n - \gamma_n, \\
    \gamma_n' + 2 V \beta_n
  \end{pmatrix}
  h_{2n}.
\end{align*}
Consequently, the coefficients must satisfy the equations
\begin{subequations}\label{eq:coeffs-diffeq}\begin{align}
  \alpha_j' &= 2 \beta_j, \label{eq:coeffs-diffeq-alpha} \\
  \beta_j' &= -\alpha_{j+1} - V \alpha_j + \gamma_j, \label{eq:coeffs-diffeq-beta} \\
  \gamma_j' &= -2 \beta_{j+1} - 2 V \beta_j. \label{eq:coeffs-diffeq-gamma}
\end{align}\end{subequations}
Moreover, the left-hand side of~\eqref{eq:H-pure} evaluates to
\begin{align*}\MoveEqLeft
  \widehat{\mathcal{H}}_{\varphi\varphi,n} \widehat{\mathcal{H}}_{\pi\pi,n} - \widehat{\mathcal{H}}_{(\varphi\pi),n}^{\,2} \\
  &= \alpha_0 \gamma_{-1} + \sum_{j=1}^n \Bigl( \alpha_0 \gamma_j + \alpha_{j+1} \gamma_{-1} - \beta_0 \beta_j + \sum_{i=1}^j (\alpha_i \gamma_{j-i} - \beta_i \beta_{j-i}) \Bigr) k_+^{-2j} + \mathcal{O}\bigl(k_+^{-2(n+1)}\bigr) \\
  &= \frac14 + \frac12 \sum_{l=1}^n \Bigl( \alpha_{j+1} + \gamma_j + 2 \sum_{i=1}^j (\alpha_i \gamma_{j-i} - \beta_i \beta_{j-i}) \Bigr) k_+^{-2j} + \mathcal{O}\bigl(k_+^{-2(n+1)}\bigr).
\end{align*}
Therefore, the coefficients must additionally satisfy the constraint
\begin{equation}\label{eq:coeffs-pure}
  \alpha_{j+1} + \gamma_j = -2 \sum_{i=1}^j (\alpha_i \gamma_{j-i} - \beta_i \beta_{j-i}).
\end{equation}

It is remarkable that (with this constraint) the differential equations~\eqref{eq:coeffs-diffeq} can be solved recursively without solving any integrals:
\begin{proposition}
  The differential equations~\eqref{eq:coeffs-diffeq} with initial values~\eqref{eq:H-init-coeffs} and constraint~\eqref{eq:coeffs-pure} are solved by the recurrence relations
  \begin{subequations}\label{eq:coeffs-recur}\begin{align}
    \alpha_{j+1} &= -\frac12 (V \alpha_j + \beta_j') - \sum_{i=1}^j (\alpha_i \gamma_{j-i} - \beta_i \beta_{j-i}), \label{eq:coeffs-recur-alpha} \\
    \beta_{j+1} &= -\frac14 V' \alpha_j - \frac14 \beta_j'' - V \beta_j, \label{eq:coeffs-recur-beta} \\
    \gamma_j &= \frac12 (V \alpha_j + \beta_j') - \sum_{i=1}^j (\alpha_i \gamma_{j-i} - \beta_i \beta_{j-i}). \label{eq:coeffs-recur-gamma}
  \end{align}\end{subequations}
\end{proposition}
\begin{proof}
  To find~\eqref{eq:coeffs-recur-beta}, we eliminate $\gamma_j'$ from~\eqref{eq:coeffs-diffeq-gamma} by adding the derivative of~\eqref{eq:coeffs-diffeq-beta} and using~\eqref{eq:coeffs-diffeq-alpha}.
  The other two equations obtained from~\eqref{eq:coeffs-pure} and $-\alpha_{j+1} + \gamma_j = V \alpha_j + \beta_j'$, which follows from~\eqref{eq:coeffs-diffeq-beta}, thus yield~\eqref{eq:coeffs-recur-alpha} and~\eqref{eq:coeffs-recur-gamma}.

  To see that~\eqref{eq:coeffs-pure} is consistent with~\eqref{eq:coeffs-diffeq}, we first subtract~\eqref{eq:coeffs-diffeq-alpha} and~\eqref{eq:coeffs-diffeq-gamma} to obtain $\alpha_{j+1}' + \gamma_j' = -2 V \beta_j$.
  Then we calculate
  \begin{align*}
    \partial_\tau \sum_{i=1}^j (\alpha_i \gamma_{j-i} - \beta_i \beta_{j-i})
    &= \sum_{i=1}^j \bigl( \alpha_i' \gamma_{j-i} + \alpha_i \gamma_{j-i}' - 2 \beta_i' \beta_{j-i} \bigr) \\
    &= 2 \sum_{i=1}^j \bigl( \beta_i \gamma_{j-i} - \alpha_i (\beta_{j-i+1} + V \beta_{j-i}) - (\gamma_i - \alpha_{i+1} - V \alpha_i) \beta_{j-i} \bigr) \\
    &= 2 \sum_{i=1}^j (\alpha_i \beta_{j-i+1} - \alpha_{i+1} \beta_{j-i}) + 2 \gamma_0 \beta_j 
    = 2 (\gamma_0 - \alpha_1) \beta_j = V \beta_j,
  \end{align*}
  where, in the last step, we used that~\eqref{eq:coeffs-diffeq-beta} for $j=0$ implies $\gamma_0 - \alpha_1 = \frac12 V$.
\end{proof}

\subsection{Comparison to Hadamard point-splitting}

The regularization procedure proposed in the previous subsection is equivalent to the typically considered Hadamard point-splitting method, in which a truncation of the Hadamard parametrix is subtracted~\cite{christensen:1,christensen:2}.

The truncated Hadamard parametrix (at order~$n$ and scale~$\lambda$) for the Klein--Gordon equation~\eqref{eq:klein-gordon} is defined as~\cite{fewster-smith,moretti:stress-energy}
\begin{align*}
  H_n(x,y) \defn \lim_{\varepsilon \to 0^+} \frac{1}{8\uppi^2} \biggl( \frac{\Delta(x,y)^{\frac12}}{\sigma_\varepsilon(x,y)} + \sum_{j=0}^{n-1} v_j(x,y) \sigma(x,y)^j \log\frac{\sigma_\varepsilon(x,y)}{\lambda^2} \biggr),
\end{align*}
where $x,y$ are in a geodesically convex neighbourhood,
\begin{equation*}
  \sigma_{\varepsilon}(x,y) \defn \sigma(x,y) + \im \varepsilon\, \bigl(t(x) - t(y)\bigr) + \frac12 \varepsilon^2,
  \quad
  \varepsilon > 0,
\end{equation*}
is the regularized Synge world function (i.e., half the signed squared geodesic distance) for a time-function~$t$, $\Delta(x,y)$ is the van Vleck--Morette determinant, $v_j(x,y)$ are smooth coefficient functions satisfying the recurrence relations~\cite{moretti:stress-energy,fewster-smith}
\begin{subequations}\label{eq:hadamard-recurrence}\begin{align}
  (K \otimes \one) \Delta^\frac12 &= \bigl( -(\Box \otimes \one) \sigma + \sigma (\Box \otimes \one) - 2 \bigr) v_0, \\
  (K \otimes \one) v_{j-1} &= \bigl( -(\Box \otimes \one) \sigma + \sigma (\Box \otimes \one) + 2j - 2 \bigr) j v_j
\end{align}\end{subequations}
for $j \in \NN$.
These relations are obtained by demanding that the Hadamard parametrix satisfies the Klein--Gordon equation up to an error of order $\sigma^{n+1}$.

Note that the coefficient functions~$v_j$ are entirely determined by the local geometry.
For example, at lowest order $v_1$ is given by
\begin{equation}\label{eq:v1}
  [v_1] = \frac{1}{8} m^4 + \frac{6\xi-1}{24} m^2 R + \frac{(6\xi-1)^2}{288} R^2 - \frac{1}{720} R_{\mu\nu} R^{\mu\nu} + \frac{1}{720} R_{\mu\nu\rho\sigma} R^{\mu\nu\rho\sigma} + \frac{5\xi-1}{120} \Box R.
\end{equation}

Since a FLRW spacetime is spatially isotropic and homogeneous, and the Hadamard parametrix is constructed from local geometric quantities, it inherits these properties so that
\begin{equation*}
  H_n\bigl((\tau,\vec{x}), (\eta,\vec{y})\bigr) = H_n\bigl(\tau, \eta, r=\abs{\vec{x}-\vec{y}}\bigr),
\end{equation*}
analogously to~\eqref{eq:omega2-homogeneous}.
(Naturally, the same holds true for $\sigma$, $\Delta^\frac12$ and $v_j$.)
As before for the two-point function, this implies that the Hadamard parametrix can be represented by
\begin{equation*}
  \tilde{\mathcal{H}}_n(\tau,r)
  \defn
  \begin{pmatrix}
    \tilde{\mathcal{H}}_{\varphi\varphi,n}(\tau,r) \\
    \tilde{\mathcal{H}}_{(\varphi\pi),n}(\tau,r) \\
    \tilde{\mathcal{H}}_{\pi\pi,n}(\tau,r)
  \end{pmatrix}
  \defn
  \lim_{\eta \to \tau}
  \begin{pmatrix}
    \one \\
    \frac12 (\partial_\tau + \partial_\eta) \\
    \partial_\tau \partial_\eta
  \end{pmatrix}
  a(\tau) a(\eta) H_n(\tau,\eta,r).
\end{equation*}

It follows from results in~\cite{eltzner-gottschalk}, that the most singular terms of~$\tilde{\mathcal{H}}_n$ and~$\mathcal{H}_n$ agree, thus justifying the choice~\eqref{eq:H-init-coeffs} of the lowest order coefficients.
(The results of~\cite{eltzner-gottschalk} are stated for cosmological time but a translation to conformal time is not difficult.)
Equivalence of $\tilde{\mathcal{H}}$ and $\mathcal{H}$ to arbitrary order then follows from the fact that both are approximate solutions of the Klein--Gordon equation with the same singular terms (of the form $r^{-j}$ and $r^j \log r$).

For $n \geq 2$, one can compute the expansions
\begin{align*}
  \begin{split}
    \tilde{\mathcal{H}}_{\varphi\varphi,n}(\tau,r) &= \frac{1}{4\uppi^2 r^2} + \frac{1}{48\uppi^2} \frac{a''(\tau)}{a(\tau)} + \mleft( \frac{V(\tau)}{16\uppi^2}  + \frac{3V(\tau)^2+V''(\tau)}{384\uppi^2} r^2 \mright) \log\mleft(\frac{r^2 a(\tau)^2}{2\lambda^2}\mright) \\&\quad + \frac{r^2}{5760\uppi^2} \mleft( 30 V(\tau) \frac{a'(\tau)}{a(\tau)} + 11 \Bigl(\partial_\tau \frac{a'(\tau)}{a(\tau)}\Bigr)^2 - 2 \frac{a''(\tau)^2}{a(\tau)^2} + 12 \frac{a'(\tau)}{a(\tau)} \partial_\tau \frac{a''(\tau)}{a(\tau)} \mright) \\&\quad + \mathcal{O}(r^4),
  \end{split} \\
  \begin{split}
    \tilde{\mathcal{H}}_{(\varphi\pi),n}(\tau,r) &= \frac{V'(\tau)}{32\uppi^2} \log\mleft(\frac{r^2 a(\tau)^2}{2\lambda^2}\mright) + \frac{1}{96\uppi^2} \mleft( 6 V(\tau) \frac{a'(\tau)}{a(\tau)} + \partial_\tau \frac{a''(\tau)}{a(\tau)} \mright) + \mathcal{O}(r^2),
  \end{split} \\
  \begin{split}
    \tilde{\mathcal{H}}_{\pi\pi,n}(\tau,r) &= -\frac{1}{2\uppi^2 r^4} + \frac{V(\tau)}{8\uppi^2 r^2} + \frac{V(\tau)^2+V''(\tau)}{64\uppi^2} \log\mleft(\frac{r^2 a(\tau)^2}{2\lambda^2}\mright) \\&\quad + \frac{1}{192\uppi^2} \mleft( 3 V(\tau)^2 - 6 V(\tau) \frac{a'(\tau)^2}{a(\tau)^2} + 4 V(\tau) \frac{a''(\tau)}{a(\tau)} + 12 V'(\tau) \frac{a'(\tau)}{a(\tau)} + V''(\tau) \mright) \\&\quad + \frac{1}{960\uppi^2} \mleft( \Bigl(\partial_\tau \frac{a'(\tau)}{a(\tau)}\Bigr)^2 + 2 \frac{a''(\tau)^2}{a(\tau)^2} + 4 \partial_\tau^2 \frac{a''(\tau)}{a(\tau)} \mright) + \mathcal{O}(r^2).
  \end{split}
\end{align*}
To obtain this result, one uses a covariant expansion of the coefficient functions $\Delta^\frac12$ and $v_j$, see e.g.~\cite{christensen:1,christensen:2,decanini-folacci}, together with an expansion of Synge's world function, see Sect.~\ref{sec:expansion}.

The formulas above suggest the convention
\begin{equation}\label{eq:mu-lambda-convention}
  \mu^2 = 2\e^{2\gamma-2} \lambda^2 \lambda_0^2,
\end{equation}
where $\gamma$ is the Euler--Mascheroni constant and $\lambda_0 > 0$ is an arbitrary dimensionless constant.
With this convention, we find the differences ($n \geq 2$)
\begin{subequations}\label{eq:reg-differences}\begin{align}
  \tilde{\mathcal{H}}_{\varphi\varphi,n} - \mathcal{H}_{\varphi\varphi,n} \Bigl|_{r=0} &= \frac{V}{8\uppi^2} \log(a \lambda_0) + \frac{1}{48\uppi^2} \frac{a''}{a}, \\
  \begin{split}
    \tilde{\mathcal{H}}_{(\varphi\pi),n} - \mathcal{H}_{(\varphi\pi),n} \Bigl|_{r=0} &= \frac{V'}{16\uppi^2} \log(a \lambda_0) + \frac{1}{96\uppi^2} \mleft( 6 V \frac{a'}{a} + \partial_\tau \frac{a''}{a} \mright),
  \end{split} \\
  \begin{split}
    \tilde{\mathcal{H}}_{\pi\pi,n} - \mathcal{H}_{\pi\pi,n} \Bigl|_{r=0} &= \frac{V^2+V''}{32\uppi^2} \log(a \lambda_0) + \frac{1}{960\uppi^2} \mleft( \Bigl(\partial_\tau \frac{a'}{a}\Bigr)^2 + 2 \frac{a^{\prime\prime\,2}}{a^2} + 4 \partial_\tau^2 \frac{a''}{a} \mright) \\&\quad + \frac{1}{192\uppi^2} \mleft( 3 V^2 - 6 V \frac{a^{\prime\,2}}{a^2} + 4 V \frac{a''}{a} + 12 V' \frac{a'}{a} + V'' \mright),
  \end{split} \\
  \begin{split}
    \upDelta_r (\tilde{\mathcal{H}}_{\varphi\varphi,n} - \mathcal{H}_{\varphi\varphi,n}) \Bigl|_{r=0} &= \frac{3V^2+V''}{32\uppi^2} \mleft( \frac{5}{6} + \log(a \lambda_0) \mright) + \frac{V}{32\uppi^2} \frac{a^{\prime\,2}}{a^2} \\&\quad + \frac{1}{960\uppi^2} \mleft( 11 \Bigl(\partial_\tau \frac{a'}{a}\Bigr)^2 - 2 \frac{a^{\prime\prime\,2}}{a^2} + 12 \frac{a'}{a} \partial_\tau \frac{a''}{a} \mright).
  \end{split}
\end{align}\end{subequations}
We will use these differences in the following section.

\section{Semiclassical Einstein equation}
\label{sec:sce}

On FLRW spacetimes, it is sufficient to look at the traced SCE
\begin{equation}\label{eq:traced-SCE}
  -R = \kappa {\langle T^\mathrm{ren} \rangle}_\omega,
\end{equation}
where $\langle T^\mathrm{ren} \rangle_\omega \defn g^{\mu\nu} \langle T_{\mu\nu}^\mathrm{ren} \rangle_\omega$ is the trace of the quantum stress-energy tensor, and at the energy-component (viz., the $00$- or $tt$-component)
\begin{equation}\label{eq:constraint-SCE}
  G_{00} = \kappa \langle T_{00}^\mathrm{ren} \rangle_\omega
\end{equation}
The latter equation can be regarded as a constraint equation, which, if imposed at a fixed time, holds everywhere because of the covariant conservation $\nabla^\mu \langle T_{\mu\nu}^\mathrm{ren} \rangle_\omega = 0$ of the stress-energy tensor~\cite{moretti:stress-energy,hollands-wald:stress-energy}, see also Sect.~\ref{sub:energy}.

We will focus most of our attention on the traced equation~\eqref{eq:traced-SCE}, which completely determines the dynamics of the geometry, given by the scale factor.
However, also the energy equation plays an important role; it is relevant when imposing consistent initial conditions.

\subsection{Renormalized stress-energy tensor}

The (classical) stress-energy tensor for a (real) free scalar field is
\begin{equation}\label{eq:stress-energy-classical}\begin{split}
  T_{\mu\nu} &= (1-2\xi) (\nabla_{\!\mu} \phi) (\nabla_{\!\nu} \phi) - \frac12 (1-4\xi) g_{\mu\nu} (\nabla^{\sigma} \phi) (\nabla_{\!\sigma} \phi) - \frac12 g_{\mu\nu} m^2 \phi^2 \\&\quad + \xi \bigr(G_{\mu\nu} \phi^2 - 2 \phi \nabla_{\!\mu} \nabla_{\!\nu} \phi - 2 g_{\mu\nu} \phi \Box \phi \bigr).
\end{split}\end{equation}
To quantize this expression, we replace products of (derivatives of) the classical fields by their renormalized quantum counterparts.
That is, using Hadamard point-splitting, we have for the expectation value of the stress-energy tensor in a state~$\omega$ (cf.\ \cite{moretti:stress-energy,siemssen:phd,hollands-wald:stress-energy}):
\begin{equation}\label{eq:stress-energy-quantum}\begin{split}
  {\langle T_{\mu\nu}^\mathrm{ren} \rangle}_\omega &= (1-2\xi) [(\nabla_{\!\mu} \otimes \nabla_{\!\nu}) \omega_2^\mathrm{reg}] - \frac12 (1-4\xi) g_{\mu\nu} [(\nabla^\sigma \otimes \nabla_{\!\sigma}) \omega_2^\mathrm{reg}] - \frac12 g_{\mu\nu} m^2 [\omega_2^\mathrm{reg}] \\&\quad + \xi \bigl( G_{\mu\nu} [\omega_2^\mathrm{reg}] - 2 [(\one \otimes \nabla_{\!\mu} \nabla_{\!\nu}) \omega_2^\mathrm{reg}] - 2 g_{\mu\nu} [(\one \otimes \Box) \omega_2^\mathrm{reg}] \bigr) \\&\quad + \frac{1}{4\uppi^2} g_{\mu\nu} [v_1] + c_1 m^4 g_{\mu\nu} + c_2 m^2 G_{\mu\nu} + c_3 I_{\mu\nu} + c_4 J_{\mu\nu},
\end{split}\end{equation}
where, for fixed but arbitrary $n \geq 1$, we set $\omega_2^\mathrm{reg} \defn \omega_2 - H_n$, $[\,\cdot\,]$ denotes the coincidence limit (e.g., $[v_1](x) = v_1(x,x)$ is the coincidence limit of the Hadamard coefficient~$v_1$) with implicit parallel transport, $c_\bullet$ are renormalization parameters, and $I_{\mu\nu}$, $J_{\mu\nu}$ are the two fourth order conserved curvature tensors:
\begin{align*}
  I_{\mu\nu} &\defn 2 R R_{\mu\nu} - 2 \nabla_{\!\mu} \nabla_{\!\nu} R - \frac12 g_{\mu\nu} (R^2 + 4 \Box R), \\
  J_{\mu\nu} &\defn 2 R^{\rho\sigma} R_{\rho \mu \sigma \nu} - \nabla_{\!\mu} \nabla_{\!\nu} R - \Box R_{\mu\nu} - \frac12 g_{\mu\nu} (R_{\rho\sigma} R^{\rho\sigma} + \Box R).
\end{align*}
These tensors can be obtained as the variations with respect to the metric of $R^2$ and $R_{\mu\nu} R^{\mu\nu}$ \cite{wald:trace}.
For conformally flat spacetimes like FLRW, it holds $I_{\mu\nu} = 3 J_{\mu\nu}$.

The term involving the Hadamard coefficient~$v_1$ in~\eqref{eq:stress-energy-quantum} ensures that the quantized stress-energy tensor satisfies $\nabla^\mu \langle T_{\mu\nu}^\mathrm{ren} \rangle_\omega = 0$~\cite{moretti:stress-energy,hollands-wald:stress-energy}.
The term $c_1 m^4 g_{a b}$ is a renormalization of the cosmological constant, and $c_2 m^2 G_{a b}$ corresponds to a renormalization of the gravitational coupling constant $\kappa$; the remaining two renormalization terms have no classical interpretation.
A change of the scale~$\lambda$ in the Hadamard parametrix corresponds to a change of the renormalization parameters~\cite{hack:book}, and thus we are in principle at liberty to set~$\lambda$ to any value if we change the renormalization parameters accordingly.

We also remark that in case the state $\omega$ includes a non-zero one-point function (classical Klein--Gordon field), the corresponding contribution to the two point function is $\phi^\mathrm{bg}(x)\phi^\mathrm{bg}(y)$, that is, in~\eqref{eq:stress-energy-quantum} we have to replace $\omega_2^\mathrm{reg}(x,y)$ by $\omega_2^\mathrm{reg}(x,y) + \phi^\mathrm{bg}(x)\phi^\mathrm{bg}(y)$.
After performing the coincidence limit, this leads to an additional contribution given by the classical stress-energy tensor~\eqref{eq:stress-energy-classical} with $\phi$ replaced by $\phi^\mathrm{bg}$.

\subsection{Trace of the stress-energy tensor}

Taking the trace of~\eqref{eq:stress-energy-quantum} yields\footnote{Note that the factor in front of $v_1$ in the corresponding Eq.~(22) of~\cite{eltzner-gottschalk} is incorrect.}
\begin{equation}\label{eq:trace-quantum}\begin{split}
  {\langle T^\mathrm{ren} \rangle}_\omega &= \bigl( (6\xi-1) (\xi R + m^2) - m^2 \bigr) [\omega_2^\mathrm{reg}] + (6\xi-1) g^{\mu\nu} [(\nabla_{\!\mu} \otimes \nabla_{\!\nu}) \omega_2^\mathrm{reg}] \\&\quad - \frac{9\xi-2}{2\uppi^2} [v_1] + 4 c_1 m^4 - c_2 m^2 R - (6 c_3 + 2 c_4) \Box R.
\end{split}\end{equation}
where $c_\bullet$ are the same constants as above, and we used that $4 \uppi^2 [(\one \otimes K) \omega_2^\mathrm{reg}] = 3 [v_1]$.

For homogeneous states and isotropic states (i.e., satisfying \eqref{eq:omega2-homogeneous}) on FLRW spacetimes, we thus find
\begin{equation}\label{eq:trace-FLRW}\begin{split}
  {\langle T^\mathrm{ren} \rangle}_\omega &= \bigl( (6\xi-1) (\xi R + m^2) - m^2 \bigr) [\omega_2^\mathrm{reg}] - \frac{6\xi-1}{a^2} \bigl( [(\upDelta \otimes \one) \omega_2^\mathrm{reg}] + [(\partial_\tau \otimes \partial_\tau) \omega_2^\mathrm{reg}] \bigr) \\&\quad - \frac{9\xi-2}{2\uppi^2} [v_1] + 4 c_1 m^4 - c_2 m^2 R - (6 c_3 + 2 c_4) \Box R,
\end{split}\end{equation}
where
\begin{equation*}
  R = 6 \frac{a''}{a^3},
  \quad
  \Box R = 36 \frac{a'' a^{\prime\,2}}{a^7} - 18 \frac{a^{\prime\prime\,2}}{a^6} - 24 \frac{a^{(3)} a'}{a^6} + 6 \frac{a^{(4)}}{a^5}
\end{equation*}
and \eqref{eq:v1} specializes to
\begin{equation*}\begin{split}
  [v_1] &= \frac{m^4}{8} + \frac{1}{60} \mleft( \frac{a^{\prime\,4}}{a^8} - \frac{a'' a^{\prime\,2}}{a^7} \mright) + \frac{(6\xi-1) m^2}{4} \frac{a''}{a^3} + \frac{(6\xi-1)^2}{8} \frac{a^{\prime\prime\,2}}{a^6} \\&\quad + \frac{5\xi-1}{20} \mleft( 6 \frac{a'' a^{\prime\,2}}{a^7} - 3 \frac{a^{\prime\prime\,2}}{a^6} - 4 \frac{a^{(3)} a'}{a^6} + \frac{a^{(4)}}{a^5} \mright).
\end{split}\end{equation*}

In the case of a non-vanishing one-point function yielding the classical field $\phi^\mathrm{bg}$, the following expression needs to be added to $\langle T^\mathrm{ren}\rangle_\omega$:
\begin{equation}\label{eq:trace-classical}
  \bigl( (6\xi-1) (\xi R + m^2) - m^2 \bigr) \bigl(\phi^\mathrm{bg}\bigr)^2 - \frac{6\xi-1}{a^2} \bigl(\dot\phi^\mathrm{bg}\bigr)^2.
\end{equation}

\subsection{Energy component of the stress-energy tensor}
\label{sub:energy}

Next, let us state the energy component of the stress-energy tensor in a homogeneous and isotropic state on an FLRW spacetime:
\begin{equation}\label{eq:T00-FLRW}\begin{split}
  {\langle T_{00}^\mathrm{ren} \rangle}_\omega &= \frac12 [(\partial_\tau \otimes \partial_\tau) \omega_2^\mathrm{reg}] - \frac12 [(\one \otimes \upDelta) \omega_2^\mathrm{reg}] + \frac12 a^2 m^2 [\omega_2^\mathrm{reg}] \\&\quad + \xi \biggl( G_{00} [\omega_2^\mathrm{reg}] + 6 \frac{a'}{a} [(\one \otimes \partial_\tau) \omega_2^\mathrm{reg}] \biggr) \\&\quad - \frac{a^2}{4\uppi^2} [v_1] - c_1 a^2 m^4 + c_2 m^2 G_{00} + (3c_3 + c_4) J_{00},
\end{split}\end{equation}
where
\begin{equation*}
  G_{00} = 3 \frac{a^{\prime\,2}}{a^2},
  \quad
  J_{00} = -24 \frac{a'' a^{\prime\,2}}{a^5} - 6 \frac{a^{\prime\prime\,2}}{a^4} + 12 \frac{a^{(3)} a'}{a^4}.
\end{equation*}

A straightforward calculation shows
\begin{equation}\label{eq:R-from-G00}
  R = \frac{1}{a^2} \mleft( \frac{a}{a'} \partial_\tau G_{00} + 2 G_{00} \mright).
\end{equation}
Consequently we also have
\begin{equation}\label{eq:T-from-T_00}
  {\langle T^\mathrm{ren} \rangle}_\omega = -\frac{1}{a^2} \biggl( \frac{a}{a'} \partial_\tau {\langle T_{00}^\mathrm{ren} \rangle}_\omega + 2 {\langle T_{00}^\mathrm{ren} \rangle}_\omega \biggr),
\end{equation}
that is, we can express the trace in terms of the energy component if $a'$ is bounded away from zero.
If \eqref{eq:traced-SCE} holds, then \eqref{eq:R-from-G00} and \eqref{eq:T-from-T_00} imply that
\begin{equation}\label{eq:G00-T00}
  \partial_\tau \bigl( G_{00} - \kappa {\langle T_{00}^\mathrm{ren} \rangle}_\omega \bigr) = -2 \frac{a'}{a} \bigl( G_{00} - \kappa {\langle T_{00}^\mathrm{ren} \rangle}_\omega \bigr)
\end{equation}
holds.
In particular, this implies that the constraint equation \eqref{eq:constraint-SCE} given by the energy component of the SCE is satisfied everywhere if it is satisfied at one instance of time.
Note also that \eqref{eq:G00-T00} implies that $G_{00} - \kappa {\langle T_{00}^\mathrm{ren} \rangle}_\omega$ is asymptotically pushed to zero if $a' > 0$.

If a non-vanishing one-point function is present, it contributes to $\langle T_{00}^\mathrm{ren} \rangle_\omega$ the following expression:
\begin{equation*}
  \frac12 \bigl(\partial_\tau \phi^{\mathrm{bg}}\bigr)^2 + \frac12 a^2 m^2 \bigl(\phi^\mathrm{bg}\bigr)^2 + \xi \biggl( G_{00} \bigl(\phi^\mathrm{bg}\bigr)^2 + 3 \frac{a'}{a} \partial_\tau \bigl(\phi^\mathrm{bg}\bigr)^2 \biggr).
\end{equation*}

\subsection{Problem of higher derivatives}

Note that the SCE contains, via the renormalized stress-energy tensor, up to fourth order derivatives of the metric and thus, in the case of a FLRW spacetime, fourth order derivatives of the scale factor.
This is in striking contrast to the classical Einstein equation, which contains only second order derivatives of the metric.
Therefore, it can be argued that the SCE has solutions which diverge significantly from solutions with similar initial data for the classical Einstein equation, see e.g.\ the discussion in~\cite{flanagan-wald}.
In particular, it is generally expected that these higher derivatives cause additional unphysical runaway solutions to appear.

A na\"ive strategy to solve the SCE for FLRW spacetimes, i.e., the equation $-R = \langle T^\mathrm{ren} \rangle_\omega$ coupled with the dynamics of the quantum state, would be to move the highest order derivatives of the scale factor to the left-hand side and all remaining terms to the right-hand side.
Unfortunately this is not possible (in the case of non-conformal coupling) because the regularization includes singular terms with fourth order derivatives of the scale factor, and thus it is not clear how to proceed, cf.~\cite{suen:coupling} where the same problem is discussed.
For this reason, earlier approaches to solving the SCE, e.g.~\cite{pinamonti-siemssen}, focused on the conformally coupled case, where this problem can be avoided.

Actually, if one wishes to solve the SCE such that the state is Hadamard, the problem is even more severe.
In that case there is no obvious possibility to specify initial values for the state such that it is Hadamard as long as the spacetime is not known in a neighbourhood of the Cauchy surface.
In the `nicest' scenario, the analytic case, one needs to know \emph{all} the derivatives of metric (here, the scale factor) at the Cauchy surface.
A possible strategy to avoid this problem is to work with adiabatic states, as done e.g.\ in~\cite{pinamonti-siemssen}, but it comes at the cost that the solutions cannot be guaranteed to be smooth.

Below we suggest an alternative approach which relies on a dynamical system for the regularized two-point function and its derivatives in the coincidence limit and thus avoids the above mentioned problems.

\subsection{Stability of solutions}

There are various notions of stability for solutions of differential equations, the most basic but also the most important being the continuous dependence of the solution on its initial data and the parameters\footnote{The notion of parameters can be understood in a very general sense and can encompass `abstract' parameters in some Banach space (or an even more general space).} in the differential equation.
Together with existence and uniqueness of solutions, continuous dependence on the initial data is required for a well-posed problem in the sense of Hadamard.

A differential equation can typically be rewritten in many more or less equivalent forms.
In the case of the SCE on FLRW spacetimes, we have seen in \eqref{eq:T-from-T_00} that the trace can be expressed in terms of the energy component.
In the derivation of that equation we divided by $a'$, and thus it is only valid if $a'$ is bounded away from zero.
Hence, solving SCE in terms of the energy equation leads to problems in the neighbourhood of Minkowski spacetime where $a' = 0$.
This fact is related to the failure of well-posedness of the energy equation near the Minkowski solution described e.g.\ in~\cite{suen:stability1,suen:stability2} as we will also discuss in Sect.~\ref{sub:solved-traced-SCE}.
Since this problem does not appear if the trace equation is taken as the governing dynamical equation, it should be preferred over the energy equation in the generic case.

\section{Dynamical system for sequences of coincidence limits}
\label{sec:moments}

\subsection{Dynamics}

For $n \in \NN_0$, we define
\begin{equation}\label{eq:moments}
  \mathcal{M}_n \defn \upDelta_r^n (\mathcal{G} - \mathcal{H}_l) \Big|_{r=0},
  \quad
  l \geq n+1.
\end{equation}
That this definition is indeed independent of~$l$ follows immediately from~\eqref{eq:homog_h} and the definition of~$\mathcal{H}_n$:

\begin{proposition}
  For $j \geq l \geq n+1$, we have
  \begin{equation*}
    \upDelta_r^n (\mathcal{H}_j - \mathcal{H}_l) \Big|_{r=0} = 0.
  \end{equation*}
\end{proposition}

Occasionally, we call $\mathcal{M}_n$ \emph{moments}.
To understand this nomenclature, consider the momentum space representation of $\mathcal{M}_n$:
\begin{equation*}
  \mathcal{M}_n = (-1)^n \int_0^\infty k^{2(n+1)} (\hat{\mathcal{G}} - \hat{\mathcal{H}}_l) \dif k.
\end{equation*}
That is, the sequence of $\mathcal{M}_n$ is given by the moments of $\hat{\mathcal{G}} - \hat{\mathcal{H}}$.
Note, however, that $\hat{\mathcal{G}} - \hat{\mathcal{H}}$ cannot be expected to be positive.

To formulate a dynamics for $\mathcal{M}_n$, recall~\eqref{eq:G-dynamics} and~\eqref{eq:H-dynamics}.
It follows that
\begin{equation*}
  \partial_\tau \mathcal{M}_n
  = \partial_\tau \upDelta_r^n (\mathcal{G} - \mathcal{H}_l) \Big|_{r=0}
  = A \mathcal{M}_n + B \mathcal{M}_{n+1},
\end{equation*}
where we defined
\begin{equation*}
  A \defn
  \begin{pmatrix*}
    0 & 2 & 0 \\
    -V & 0 & 1 \\
    0 & -2V & 0
  \end{pmatrix*},
  \quad
  B \defn
  \begin{pmatrix*}
    0 & 0 & 0 \\
    1 & 0 & 0 \\
    0 & 2 & 0
  \end{pmatrix*}.
\end{equation*}
Sometimes we write $A(\tau)$ to emphasize the dependence on~$\tau$ for a fixed potential~$V$.
Henceforth we shall suppose that $V$ is an arbitrary (but sufficiently regular) function and can forget its relation to the Klein--Gordon equation.

Considering $\mathcal{M}_n$ as a sequence $\mathcal{M} = (\mathcal{M}_n)$, we can also write the equation above as
\begin{equation}\label{eq:moment-dynamics}
  \partial_\tau \mathcal{M}(\tau) = \bigl(A(\tau) \otimes \one + B \otimes L\bigr) \mathcal{M}(\tau),
\end{equation}
where $L$ denotes the left-shift operator.
Below we study this equation on the weighted sequence spaces
\begin{equation}\label{eq:vec-ell}
  \vec\ell^p(w) \defn \RR^3 \otimes \ell^p(w),
\end{equation}
where $p \geq 1$ and $w$ is a sequence of weights; we denote the norms by ${\norm{\,\cdot\,}}_{p,w}$.
See Sect.~\ref{sec:weighted-spaces} for an introduction and our conventions.
Note in particular that in our convention the weights appear as inverses in the norm and thus they directly translate to the maximum growth rate of elements in the sequence space.

The careful reader will have noticed that above we have essentially performed a Taylor expansion of $\mathcal{G} - \mathcal{H}$ in the radial variable $r$.
Indeed, as we will see below, our approach can be applied in cases where that difference is analytic, with the best result in the entire case.
Therefore our resulting reformulation of the initial value problem for the SCE can partially be solved with methods related to the Cauchy--Kovalevskaya theorem.
However, unlike in the classical Cauchy--Kovalevskaya theorem, we do not require analyticity in the time variable (continuous differentiability of sufficient order is enough) and thus our methods more closely resemble the more abstract ones by Ovsyannikov \cite{ovsyannikov,friesen}.
For the same reason, the subset of moments that we consider cannot be compared with the set of analytic Hadamard states (i.e., Hadamard states which satisfy the analytic wavefront set condition) \cite{strohmaier-verch-wollenberg} except under additional assumptions including the analyticity of the scale factor.
However, one can expect that requiring analyticity also in time would be too restrictive for most physically interesting solutions.

In the following two subsections, we calculate $\mathcal{M}$ in two relevant examples to motivate the choice of weights later on.
Afterwards, the remainder of this section is concerned with solving~\eqref{eq:moment-dynamics}.
The solution is given by a mathematically rigorous definition of the time-ordered exponential of $A \otimes \one + B \otimes L$.

Finally, we wish to remark that the set of possible sequences $\mathcal{M}$ contains many unphysical examples (in particular, positivity is not guaranteed), and, furthermore, not all possible Hadamard states can be represented as a sequence in~\eqref{eq:vec-ell} for the weights considered below.

\subsection{\texorpdfstring{$\mathcal{M}$}{M} for the massive vacuum state on Minkowski spacetime}
\label{sec:vacuum}

First, let us recall the expansions of the modified Bessel functions, see e.g.\ Chap.~10 of~\cite{nist}:
\begin{align*}
  I_\nu(z) &= \bigl(\tfrac12 z\bigr)^\nu \sum_{n=0}^\infty \frac{\bigl(\tfrac14 z^2\bigr)^j}{j!\Gamma(\nu+j+1)}, \\
  K_n(z) &= \tfrac12 \bigl(\tfrac12 z\bigr)^{-n} \sum_{j=0}^{n-1} \frac{(n-j-1)!}{j!} \bigl(-\tfrac14 z^2\bigr)^j + (-1)^{n+1} \log\bigl(\tfrac12 z\bigr) I_n(z) \\&\quad + (-1)^n \tfrac12 \bigl(\tfrac12 z\bigr)^n \sum_{j=0}^\infty \bigl(\psi(j+1) + \psi(n+j+1)\bigr) \frac{\bigl(\tfrac14 z^2\bigr)^j}{j!(n+j)!},
\end{align*}
for $\nu \in \CC$ and $n = 0,1,2,\dotsc$, where $\psi$ denotes the Digamma function.

The two-point function for the massive vacuum state on Minkowski spacetime has the initial data:
\begin{subequations}\label{eq:G-vacuum}\begin{align}
  \mathcal{G}_{\varphi\varphi}(r) &= \frac{m}{4\uppi^2r} K_1(mr) \nonumber\\&= \frac{1}{4\uppi^2r^2} + \frac{m^2}{8\uppi^2} \sum_{j=0}^\infty \Bigl( \log(\tfrac12 mr) - \tfrac12 \bigl( \psi(j+1) + \psi(j+2) \bigr)\Bigr) \frac{\bigl(\tfrac12 mr\bigr)^{2j}}{j!(j+1)!}, \\
  \mathcal{G}_{(\varphi\pi)}(r) &= 0, \\
  \mathcal{G}_{\pi\pi}(r) &= -\frac{m^2}{4\uppi^2r^2} K_2(mr) \nonumber\\
  \begin{split}
    &= -\frac{1}{2\uppi^2r^4} + \frac{m^2}{8\uppi^2r^2} \\&\quad + \frac{m^4}{16\uppi^2} \sum_{j=0}^\infty \Bigl( \log(\tfrac12 mr) - \tfrac12 \bigl( \psi(j+1) + \psi(j+3) \bigr)\Bigr) \frac{\bigl(\tfrac12 mr\bigr)^{2j}}{j!(j+2)!},
  \end{split}
\end{align}\end{subequations}

Now we determine $\mathcal{H}_n$ on Minkowski spacetime.
For this purpose, we need to solve the recurrence relations~\eqref{eq:coeffs-recur} for constant $V = m^2$.
Hence, also the coefficients $\alpha_\bullet, \beta_\bullet$ and $\gamma_\bullet$ are constant, and~\eqref{eq:coeffs-recur} imply $\beta_\bullet = 0$,
\begin{equation}\label{eq:alpha+-gamma-Minkowski}
  \alpha_{j+1} - \gamma_j = -m^2 \alpha_j,
  \quad
  \alpha_{j+1} + \gamma_j = -2\sum_{i=1}^j \alpha_i \gamma_{j-i}.
\end{equation}
This recurrence relation can be solved in closed form:

\begin{proposition}
  If $V = m^2$ is constant,
  \begin{equation}\label{eq:coeffs-Minkowski}
    \alpha_j = \frac12 \binom{-\frac12}{j} m^{2j},
    \quad
    \beta_j = 0,
    \quad
    \gamma_{j-1} = \frac12 \binom{\frac12}{j} m^{2j}.
  \end{equation}
\end{proposition}
\begin{proof}
  Inserting~\eqref{eq:coeffs-Minkowski} on the left-hand side of~\eqref{eq:alpha+-gamma-Minkowski} yields
  \begin{align*}
    \alpha_{j+1} - \gamma_j &= \frac12 \biggl( \binom{-\frac12}{j+1} - \binom{\frac12}{j+1} \biggr) m^{2(j+1)} = -\frac12 \binom{-\frac12}{j} m^{2(j+1)} = -m^2 \alpha_j, \\
    \alpha_{j+1} + \gamma_j &= \frac12 \biggl( \binom{-\frac12}{j+1} + \binom{\frac12}{j+1} \biggr) m^{2(j+1)} = \frac12 \biggl( -\sum_{i=0}^{j+1} \binom{-\frac12}{i} \binom{\frac12}{j-i+1} + \binom{\frac12}{j+1} \biggr) m^{2(j+1)} \\&= -\frac12 m^{2(j+1)} \sum_{i=1}^j \binom{-\frac12}{i} \binom{\frac12}{j-i+1} = -2 \sum_{i=1}^j \alpha_i \gamma_{j-i},
  \end{align*}
  where in the second step of the second equation we used the Chu--Vandermonde identity.
\end{proof}

Then, with the coefficients~\eqref{eq:coeffs-Minkowski}, we find
\begin{subequations}\label{eq:H-vacuum}\begin{align}
  \mathcal{H}_{\varphi\varphi,n}(r) &= \frac{1}{4\uppi^2r^2} + \frac{m^2}{8\uppi^2} \sum_{j=0}^{n-1} \biggl( \log\Bigl(\frac{r}{\mu}\Bigr) - \psi(2j+2) \biggr) \frac{\bigl(\tfrac12 mr\bigr)^{2j}}{j!(j+1)!}, \\
  \mathcal{H}_{(\varphi\pi),n}(r) &= 0, \\
  \mathcal{H}_{\pi\pi,n}(r) &= -\frac{1}{2\uppi^2r^4} + \frac{m^2}{8\uppi^2r^2} + \frac{m^4}{16\uppi^2} \sum_{j=0}^{n-1} \biggl( \log\Bigl(\frac{r}{\mu}\Bigr) - \psi(2j+2) \biggr) \frac{\bigl(\tfrac12 mr\bigr)^{2j}}{j!(j+2)!}.
\end{align}\end{subequations}
Therefore, subtracting~\eqref{eq:H-vacuum} (of sufficiently high order) from~\eqref{eq:G-vacuum} and differentiating $n$~times with~$\upDelta$, we obtain in the coinciding point limit $r \to 0$:
\begin{subequations}\label{eq:M-vacuum}\begin{align}
  \mathcal{M}_{\varphi\varphi,n} &= \frac{1}{2\uppi^2} \bigl(\tfrac12 m\bigr)^{2n+2} \Bigl( \log(\tfrac12 m \mu) + \psi(2n+2) - \tfrac12 \bigl( \psi(n+1) + \psi(n+2) \bigr)\Bigr) \binom{2n+1}{n+1}, \\
  \mathcal{M}_{(\varphi\pi),n} &= 0, \\
  \mathcal{M}_{\pi\pi,n} &= \frac{1}{\uppi^2} \bigl(\tfrac12 m\bigr)^{2n+4} \Bigl( \log(\tfrac12 m \mu) + \psi(2n+2) - \tfrac12 \bigl( \psi(n+1) + \psi(n+3) \bigr)\Bigr) \frac{(2n+1)!}{n!(n+2)!}.
\end{align}\end{subequations}

Let $r \geq \frac32$.
It follows from the properties of the Digamma function (e.g., (5.4.14) and (5.4.15) of~\cite{nist}) that
\begin{align*}
  \psi(2n+2) - \frac12 \bigl( \psi(n+1) + \psi(n+r) \bigr) \leq \log(2)
\end{align*}
and the bound is approached as $n \to \infty$.
Moreover, by the duplication formula for the Gamma function,
\begin{equation*}
  \frac{(2n+1)}{n!(n+r)!} \leq \frac{2^{2n+1}}{\sqrt\uppi}.
\end{equation*}
Consequently, we have
\begin{equation*}
  \abs{\mathcal{M}_{\varphi\varphi,n}} \leq \frac{\abs{\log(m \mu)}}{4\uppi^{5/2}} m^{2n+2},
  \quad
  \abs{\mathcal{M}_{\pi\pi,n}} \leq \frac{\abs{\log(m \mu)}}{8\uppi^{5/2}} m^{2n+4},
\end{equation*}
whence $\mathcal{M} \in \vec\ell^p(w)$ for $p \geq 1$ and $w_n = \omega^n$ with $\omega > m^2$.

As a further consistency check, we show that $(A \otimes \one + B \otimes L) \mathcal{M} = 0$ for $V = m^2$.
Indeed, inserting~\eqref{eq:M-vacuum} into the left-hand side, a straightforward calculation shows that $-m^2 \mathcal{M}_{\varphi\varphi,n} + \mathcal{M}_{\pi\pi,n} + \mathcal{M}_{\varphi\varphi,n+1} = 0$.

Finally, we note that in the massless case $m=0$ all moments vanish exactly: $\mathcal{M}=0$.

\subsection{\texorpdfstring{$\mathcal{M}$}{M} for the massless thermal state on Minkowski spacetime}
\label{sec:thermal}

The two-point function for the massless KMS-state on Minkowski spacetime at inverse temperature~$\beta$ has the initial data
\begin{subequations}\label{eq:G-thermal}\begin{align}
  \mathcal{G}_{\varphi\varphi}(r) &= \frac{1}{4\uppi r\beta} \coth\Bigl(\frac{\uppi r}{\beta}\Bigr) = \frac{1}{4\uppi^2r^2} + \frac{1}{2\uppi^2\beta^2} \sum_{j=0}^\infty (-1)^j \zeta(2j+2) \Bigl(\frac{r}{\beta}\Bigr)^{2j}, \\
  \mathcal{G}_{(\varphi\pi)}(r) &= 0, \\
  \mathcal{G}_{\pi\pi}(r) &= -\frac{\uppi}{2r\beta^3} \coth\Bigl(\frac{\uppi r}{\beta}\Bigr) \csch\Bigl(\frac{\uppi r}{\beta}\Bigr)^2 \nonumber\\&= -\frac{1}{2\uppi^2r^4} + \frac{1}{2\uppi^2\beta^4} \sum_{j=0}^\infty (-1)^j (2j+2) (2j+3) \zeta(2j+4) \Bigl(\frac{r}{\beta}\Bigr)^{2j},
\end{align}\end{subequations}
where $\zeta$ denotes the (Riemann) Zeta function.
These data can be obtained after an appropriate Bogoliubov transformation from the two-point function of the vacuum state.
The coefficients for $\mathcal{H}_n$ with $V=0$ are all zero except for $\alpha_0$ and $\gamma_{-1}$.
Therefore,
\begin{subequations}\label{eq:H-thermal}\begin{align}
  \mathcal{H}_{\varphi\varphi,n}(r) = \frac{1}{4\uppi^2r^2},
  \quad
  \mathcal{H}_{(\varphi\pi),n}(r) = 0,
  \quad
  \mathcal{H}_{\pi\pi,n}(r) = -\frac{1}{2\uppi r^4}.
\end{align}\end{subequations}
Subtracting~\eqref{eq:H-thermal} from~\eqref{eq:G-thermal} and differentiating $n$ times with $\upDelta$, we obtain in the coinciding point limit $r \to 0$:
\begin{align*}
  \mathcal{M}_{\varphi\varphi,n} &= (-1)^n \beta^{-2n-2} \zeta(2n+2) (2n+1)!, \\
  \mathcal{M}_{(\varphi\pi),n} &= 0, \\
  \mathcal{M}_{\pi\pi,n} &= (-1)^n \beta^{-2n-4} \zeta(2n+4) (2n+3)!.
\end{align*}

Note that $\zeta(2) = \frac16\uppi^2$, $\zeta(4) = \frac{1}{90}\uppi^4$ and the Zeta function $\zeta(s)$ is monotonically decreasing for $s > 1$ with limit~$1$.
Hence $\mathcal{M} \in \vec\ell^p(w)$ for $p \geq 1$ and $w_n = (2n)! \omega^{2n}$ with $\omega>\beta^{-1}$.

As a further consistency check, we show that $(A \otimes \one + B \otimes L) \mathcal{M} = 0$ for $V=0$.
Indeed, it is easily checked that $\mathcal{M}_{\pi\pi,n} + \mathcal{M}_{\varphi\varphi,n+1} = 0$.

\subsection{Properties of the matrices \textit{A} and \textit{B}}

The matrix $B$ is nilpotent: $B^3 = 0$.
For this reason, many products of the matrices~$A$ and~$B$ vanish.

We are interested in the non-zero words in~$A$ and~$B$ with the largest number of $B$-factors.
It is already an easy consequence of the nilpotency of~$B$, that products of length~$3n$ contain at most $2n$ $B$-factors (e.g., $(B^2A)^n$).
This implies that a non-zero word of length~$n$ contains at most $\ceil{\frac{2n}{3}}$ $B$-factors.
This estimate can be substantially improved using further relations between the matrices $A$ and $B$.

In the following, we denote by dots the non-zero entries of a $3\times3$ matrix.
For instance, $(\dotmatrix1111)$ denotes any $3\times3$ matrix $(a_{i\!j})$ whose only non-zero entries are $a_{12}$, $a_{21}$, $a_{23}$, $a_{32}$.
The matrix $A$ introduced in the previous subsection is an example of such a matrix and we write $A = (\dotmatrix1111)$.
Another example is $B = (\dotmatrix0101)$.
Products and powers of $3\times3$ matrices will also be represented in this notation.
For example, we write $(\dotmatrix1111)^3$ for $A(t_1) A(t_2) A(t_3)$.
Note that $(\dotmatrix0101)^3 = 0$.

A quick calculation shows the following additional relations pertaining to products of the matrices~$A$ and~$B$:
\begin{equation}\label{eq:relations}\begin{alignedat}{3}
  (\dotmatrix0101)(\dotmatrix1111)(\dotmatrix0101) &= (\dotmatrix0101), \quad &
  (\dotmatrix1111)(\dotmatrix0101)(\dotmatrix0101) &= (\dotmatrix0100), \quad &
  (\dotmatrix0101)(\dotmatrix0101)(\dotmatrix1111) &= (\dotmatrix0001), \\
  (\dotmatrix1111)(\dotmatrix0101)(\dotmatrix1111) &= (\dotmatrix1111), \quad &
  (\dotmatrix0101)(\dotmatrix1111)(\dotmatrix1111) &= (\dotmatrix0111), \quad &
  (\dotmatrix1111)(\dotmatrix1111)(\dotmatrix0101) &= (\dotmatrix1101).
\end{alignedat}\end{equation}
Using these relations we can show that non-zero words in~$A$ and~$B$ of length~$n$ contain at most $\ceil*{\frac{n+1}{2}}$ $B$-factors.
More generally, we have:

\begin{proposition}\label{prop:B-maximal-products}
  A non-zero word in $(\dotmatrix1111)$ and $(\dotmatrix0101)$ of length~$n$ contains at most $\ceil*{\frac{n+1}{2}}$ $(\dotmatrix0101)$-factors.
\end{proposition}
\begin{proof}
  By induction, it follows easily from~\eqref{eq:relations}, that a non-zero word of odd length with a maximal number of $(\dotmatrix0101)$-factors has the shape $(\dotmatrix0101)$, $(\dotmatrix0100)$ or $(\dotmatrix0001)$.
  This also shows that such a $(\dotmatrix0101)$-maximal word contains at most one occurrence of two consecutive $(\dotmatrix0101)$-factors.
  Hence we find that a non-zero word of length $2n+1$ contains at most $n+1$ factors of the shape~$(\dotmatrix0101)$.
  It is then obvious, that a non-zero word of length~$2n$ can also not contain more than $n+1$ factor of the shape~$(\dotmatrix0101)$.
  This completes the proof.
\end{proof}

We conclude this subsection on the properties of the matrices~$A$ and~$B$ by noting that their matrix (max) norms are given by:
\begin{equation}\label{eq:AB-norm}
  \norm{A} = 2 \sqrt{1 + V^2},
  \quad
  \norm{B} = 2.
\end{equation}

\subsection{Evolution operator}

In this subsection, we define the evolution operator for
\begin{equation}\label{eq:S-defn}
  S(\tau) \defn A(\tau) \otimes \one + B \otimes L
\end{equation}
as an operator on the weighted sequence spaces $\vec\ell^p(w)$ (see~\eqref{eq:vec-ell} and Sect.~\ref{sec:weighted-spaces} for the definition of these spaces) for certain weight sequences~$w$ and $p \geq 1$.

Recall that evolution operator for $S(\tau)$ is the `solution operator'
\begin{equation*}
  U(\tau, \tau_0) \mathcal{M}_0 = \mathcal{M}(\tau)
\end{equation*}
associated to the solutions of the initial value problem for the dynamical equation~\eqref{eq:moment-dynamics} with initial data $\mathcal{M}_0$ at $\tau_0$:
\begin{empheq}[left=\empheqlbrace]{align*}
  \partial_\tau \mathcal{M}(\tau) &= S(\tau) \mathcal{M}(\tau), \\
  \mathcal{M}(\tau_0) &= \mathcal{M}_0.
\end{empheq}
It can formally be obtained via the time-ordered exponential
\begin{equation*}
  U(\tau,\tau_0) \defn \operatorname{Texp}\mleft( \int_{\tau_0}^\tau S(\eta) \dif\eta \mright).
\end{equation*}
That is, defining for $n>0$
\begin{equation}\label{eq:dyson-summands}
  U_n(\tau,\tau_0) \defn \int_{\tau_0}^\tau \int_{\tau_0}^{\tau_1} \!\dotsm\! \int_{\tau_0}^{\tau_{n-1}} S(\tau_1) \,\dotsm\, S(\tau_n) \dif \tau_n \dotsm \dif \tau_1,
\end{equation}
it is given by the Dyson series
\begin{equation}\label{eq:dyson-series}
  U(\tau,\tau_0) = \begin{dcases*}
    \one + \sum\nolimits_{n=1}^\infty U_n(\tau,\tau_0)        & for $\tau \geq \tau_0$ \\
    \one + \sum\nolimits_{n=1}^\infty (-1)^n U_n(\tau_0,\tau) & for $\tau < \tau_0$.
  \end{dcases*}
\end{equation}
Note that this is simply the solution of the Picard iteration method.

Motivated by the examples in Subsects.~\ref{sec:vacuum} and~\ref{sec:thermal}, in the following two subsections we show that the series~\eqref{eq:dyson-series} converges as an operator
\begin{enumerate}[label=\normalfont\tbfigures(W\arabic*)]
  \item\label{item:weights1} on $\vec\ell^p(w)$ for $w_n = c \omega^n$ with $c,\omega > 0$,
  \item\label{item:weights2} from $\vec\ell^p(w)$ to $\vec\ell^p(v)$ for $w_n = (2n)! \omega^{2n}$ and $v_n = (2n)! \upsilon^{2n}$ with $\upsilon > \omega \geq 1$ in bounded time intervals,
\end{enumerate}
and that it has the properties of an evolution operator, i.e.,
\begin{enumerate}[label=\normalfont\tbfigures(E\arabic*)]
  \item\label{item:e1} $U(\tau,\tau) = \one$,
  \item\label{item:e2} $U(\tau,\tau_0) = U(\tau,\tau_1) U(\tau_1,\tau_0)$,
  \item\label{item:e3} $\partial_\tau U(\tau,\tau_0) = S(\tau) U(\tau,\tau_0)$ and $\partial_{\tau_0} U(\tau,\tau_0) = -U(\tau,\tau_0) S(\tau_0)$
\end{enumerate}
for $\tau, \tau_0, \tau_1$ in an appropriately chosen interval.

Let us briefly remark on the relation of the spaces in (W1) and (W2) to spaces of analytic functions.
Note that $\ell^p(w) \subset \ell^\infty(w)$ for any $p \in [1,\infty]$, and hence we will restrict our attention to the case $p = \infty$.
A sequence of moments $(\mathcal{M}_n)$ is in the sequence space $\ell^\infty(w)$ of factorially growing weights $w_n = \omega^{2n} (2n)!$ \ref{item:weights2}, if there exists a constant $M > 0$ such that
\begin{equation*}
  \abs{\mathcal{M}_n} \leq M \omega^{2n} (2n)!
\end{equation*}
for all $n \in \NN_0$.
Compare this to the following well-known definition: a smooth function $f$ is analytic in $K \Subset \RR$ if and only if there exist constants $C > 0$ and $M > 0$ such that
\begin{equation*}
  \sup_{x \in K}\ \abs{f^{(n)}(x)} \leq M C^n n!.
\end{equation*}
Hence, on the one hand, the moments define a power series of an even (i.e., all non-even powers are zero) analytic function in the neighbourhood of the origin.
On the other hand, if $\mathcal{G} - \mathcal{H}$ is analytic in a neighbourhood of the origin, then the corresponding moments will be contained in an appropriate weighted sequence space of type \ref{item:weights2}.

If the moments are in a space with geometrically growing weights $w_n = \omega^{2n}$ \ref{item:weights1}, then the moments grow so slow compared to $1/(2n)!$ that the radius of convergence of the corresponding power series is infinite, and thus they define an entire analytic function.

\subsection{Evolution operator for geometrically growing weights}

In this subsection we consider the case~\ref{item:weights1}, i.e., the evolution operator on $\vec\ell^p(w)$ for the geometrically growing weights $w_n = c \omega^n$ with $c,\omega > 0$.
As described above, in this case $\mathcal{G}-\mathcal{H}$ must be entire analytic in the radial variable.
In this case, $S(\tau)$ is a bounded operator on $\vec\ell^p(w)$ with
\begin{align}\label{eq:S-bound-geom}
  {\norm{S(\tau) \mathcal{M}}}_{p,w}
  \leq (\norm{A(\tau)} + \norm{B} \omega) {\norm{\mathcal M}}_{p,w}
  \leq 2\bigl(\sqrt{1+V(\tau)^2} + \omega\bigr) {\norm{\mathcal M}}_{p,w}.
\end{align}

The following theorem is easily shown, see e.g.~Thms.~5.1 and~5.2 in~\cite{pazy} (the proofs of the analogous Thm.~\ref{thm:fact_weights} in the next subsection can also be adapted to the case of geometrically growing weights discussed here):
\begin{theorem}\label{thm:geom_weights}
  Let $c,\omega > 0$ and $p \geq 1$.
  Suppose that $I \subsetneq \RR$ and $V \in C(I)$.
  The evolution operator $U(\tau,\tau_0)$, defined by~\eqref{eq:dyson-series}, is the unique bounded operator on $\vec\ell^p(w)$ with weights $w_n = c \omega^n$, such that the properties \ref{item:e1}--\ref{item:e3} hold for $\tau,\tau_0,\tau_1 \in I$.
  Moreover, we have the bound
  \begin{equation*}
    {\norm{U(\tau,\tau_0)}}_{p,w} \leq \e^{C_\omega(\tau,\tau_0)},
  \end{equation*}
  where
  \begin{equation}\label{eq:C(tau,tau_0)}
    C_\omega(\tau,\tau_0) \defn 2 \omega \abs{\tau - \tau_0} + 2 \abs*{\int_{\tau_0}^\tau \sqrt{1+V(\eta)^2} \dif\eta},
  \end{equation}
  and the evolution operator is norm-continuous:\footnote{Together with \ref{item:e2} this implies norm-continuity of $(\tau,\sigma) \to U(\tau,\sigma).$} $\lim_{\tau \to \tau_0} {\norm{U(\tau,\tau_0) - \one}}_{p,w} = 0$.
\end{theorem}
\begin{proof}
  Let us sketch a proof.
  Using the Banach fixed point theorem and the boundedness of the operators $S(\tau)$, it can be seen that the integral equation
  \begin{equation*}
    U(\tau, \tau_0) = \one + \int_{\tau_0}^\tau S(\eta) U(\eta, \tau_0) \dif\eta
  \end{equation*}
  has a unique solution which is given the Dyson series \eqref{eq:dyson-series}.
  Property \ref{item:e1} and the first part of \ref{item:e3} are obvious, while \ref{item:e2} follows from the uniqueness of the solution.
  For the second part of \ref{item:e3} note that~\ref{item:e2} implies $U(\tau,\tau_0)^{-1} = U(\tau_0,\tau)$ and, together with the first part of~\ref{item:e3}, we thus find $\partial_{\tau_0} U(\tau,\tau_0) = -U(\tau,\tau_0) S(\tau_0)$.
  Next, the bound can be obtained via Gr{\"o}nwall's inequality from the integral equality together with the bound~\eqref{eq:S-bound-geom}.
  Finally, norm-continuity can also easily be derived from the integral equation.
\end{proof}

In particular, this theorem implies the following:
\begin{remark}\label{rem:geom-growing}
  If we are given initial data of geometrically growing moments (e.g., the moments of the vacuum state on Minkowski spacetime), they will continue to be geometrically growing under time-evolution independent of the concrete potential~$V(t)$.
\end{remark}
\begin{remark}
  Given a continuous potential $V$, the solution exists in~$\vec\ell^p(w)$ for arbitrarily large time intervals $I$.
\end{remark}

The following perturbation result can be shown by a standard application of Gr{\"o}nwall's inequality, but can also be shown along similar lines as Thm.~\ref{thm:pert-fact_weights}:
\begin{theorem}\label{thm:pert-geom_weights}
  In addition to the assumptions of Thm.~\ref{thm:geom_weights}, suppose that $\tilde{V} \in C(I)$.
  Denote by $U(\tau,\tau_0)$ and $\tilde{U}(\tau,\tau_0)$ the evolution operators for $V$ and $\tilde V$.
  If $\mathcal{M}, \tilde{\mathcal M} \in \vec\ell^p(w)$, we have
  \begin{align*}\MoveEqLeft
    {\norm{U(\tau,\tau_0) \mathcal{M} - \tilde{U}(\tau,\tau_0) \tilde{\mathcal M}}}_{p,w} \\&\leq \e^{C_\omega(\tau,\tau_0)} {\norm{\mathcal{M} - \tilde{\mathcal M}}}_{p,w} + 2 \e^{C_\omega(\tau,\tau_0)+\tilde{C}_\omega(\tau,\tau_0)} {\norm{\tilde{\mathcal M}}}_{p,w} \int_{\tau_0}^\tau \abs{V(\eta)-\tilde{V}(\eta)} \dif\eta.
  \end{align*}
\end{theorem}

Under more stringent assumptions on the potential $V$, one obtains the following regularity result:
\begin{proposition}\label{prop:geometrically-smooth}
  If, in addition to the assumptions of Thm.~\ref{thm:geom_weights}, $V \in C^k(I)$, then $\partial_\tau^{k+1} U(\tau,\tau_0)$ is bounded on $\vec\ell^p(w)$.
\end{proposition}
\begin{proof}
  By property \ref{item:e3}, $\partial_\tau U(\tau,\tau_0)$ is bounded on $\vec\ell^p(w)$.
  Repeated differentiation of \ref{item:e3}, together with the boundedness of $S(\tau)$ and its derivatives, yields the desired result by iteration.
\end{proof}

This implies, in particular, that the solutions of~\eqref{eq:moment-dynamics} with initial conditions in $\vec\ell^p(w)$ are smooth if $V$ is smooth.

\subsection{Evolution operator for factorially growing weights}

For weights growing faster than the geometric growth considered in the previous subsection, the left-shift operator and also $S(\tau)$ are unbounded (see also~\eqref{eq:L-norm}).
Moreover, the following shows that a construction based on the standard Hille--Yosida theory of $C_0$-semigroups can not be used to resolve this:
\begin{proposition}
  If the weights grow faster than geometrically, the resolvent set of $S(\tau)$ is empty.
\end{proposition}
\begin{proof}
  For any $\lambda, \xi \in \CC$,
  \begin{equation*}
    \mathcal{N}_n = \frac{1}{4^n} \bigl(4V(\tau)+\lambda^2\bigr)^n \begin{pmatrix}
      4\xi \\ 2\lambda\xi \\ \lambda^2\xi
    \end{pmatrix}
  \end{equation*}
  is in $\vec\ell^p(w)$ because it grows at most geometrically, and a short calculation shows
  \begin{equation*}
    (S(\tau) - \lambda) \mathcal{N} = \bigl((A(\tau) - \lambda) \otimes \one + B \otimes L \bigr)\mathcal{N} = 0.
  \end{equation*}
\end{proof}

However, even then it is sometimes possible to understand $S(\tau)$ and, somewhat surprisingly, also the evolution operator~\eqref{eq:dyson-series} as a bounded operator \emph{between two different weighted sequence spaces} $\vec\ell^p(v)$ and $\vec\ell^p(w)$.
In other words, both $S(\tau)$ and the evolution operator $U(\tau,\tau_0)$ can be considered as unbounded operators on~$\vec\ell^p(v)$ with domain $\vec\ell^p(w)$.

Let us consider the case~\ref{item:weights2}, i.e., the weights $w_n = (2n)! \omega^{2n}$ and $v_n = (2n)! \upsilon^{2n}$ with $\upsilon > \omega \geq 1$.
As explained earlier, this can essentially be understood as the case where $\mathcal{G}-\mathcal{H}$ is analytic (with radius of convergence related to $\omega$ and $\upsilon$) in the radial variable.
We apply Prop.~\ref{prop:B-maximal-products}, \eqref{eq:AB-norm} and~\eqref{eq:dyson-summands} to obtain for $n > 0$ the bound
\begin{align*}
  {\norm{U_n(\tau,\tau_0)\mathcal{M}}}_{p,v}
  &\leq \frac{1}{n!} \sum_{m=0}^{\ceil{\frac{n+1}{2}}} \binom{n}{m} \abs*{\int_{\tau_0}^\tau \norm{A(\eta)} \dif\eta}^{n-m} \norm{B}^m {\norm{(\one \otimes L^m) \mathcal{M}}}_{p,v} \\
  &\leq \frac{2^n}{n!} \sum_{m=0}^{\ceil{\frac{n+1}{2}}} \binom{n}{m} r_m \abs*{\int_{\tau_0}^\tau \sqrt{1+V(\eta)^2} \dif\eta}^{n-m} {\norm{\mathcal M}}_{p,w},
\end{align*}
where we set
\begin{equation*}
  r_m = \sup_n \frac{w_{n+m}}{v_n} = \sup_n \frac{(2m+2n)!}{(2n)!} \frac{\omega^{2m+2n}}{\upsilon^{2n}}.
\end{equation*}
We compute $r_0 = 1$ and, for $m > 0$,
\begin{align}
  \frac{r_m}{m!}
  &< \frac{(2m)!}{\sqrt{\uppi m} m!} \mleft(\frac\upsilon\omega\mright)^\frac12\! \mleft(\frac{\upsilon\omega}{\upsilon-\omega}\mright)^{2m}
  = \frac{\Gamma(m+\tfrac12)}{\uppi \sqrt{m}} \mleft(\frac\upsilon\omega\mright)^\frac12\! \mleft(\frac{2\upsilon\omega}{\upsilon-\omega}\mright)^{2m} \nonumber \\
  &< \frac{(m-1)!}{\uppi} \mleft(\frac\upsilon\omega\mright)^\frac12\! \mleft(\frac{2\upsilon\omega}{\upsilon-\omega}\mright)^{\mathrlap{2m}}, \label{eq:rm-bound}
\end{align}
where we applied Prop.~\ref{prop:binomial}, the duplication formula for the Gamma function, and Gautschi's inequality~\eqref{eq:gautschi}.

In Prop.~\ref{prop:ratio-sum-ineq}, we show that
\begin{equation*}
  \sum_{m=1}^{\ceil{\frac{n+1}{2}}} \frac{(m-1)!}{(n-m)!} \leq \begin{cases*}
    \frac32         & odd $n$, \\
    \frac{n}{2} + 1 & even $n$.
  \end{cases*}
\end{equation*}
Thus, estimating the remaining factors in the sum by their supremum and recalling the definition~\eqref{eq:C(tau,tau_0)} of $C_\omega(\tau,\tau_0)$, we obtain for $n > 0$
\begin{equation}\label{eq:Un-bound-factorial}
  {\norm{U_n(\tau,\tau_0)\mathcal{M}}}_{p,v} < C_0(\tau,\tau_0)^n \biggl(\frac{1}{n!} + \frac{1}{\uppi} \mleft(\frac{n}{2} + 1\mright) \mleft(\frac\upsilon\omega\mright)^\frac12\! \mleft(\frac{2\upsilon\omega}{\upsilon-\omega}\mright)^{n+2}\biggr) {\norm{\mathcal M}}_{p,w}.
\end{equation}
Consequently, the time-ordered exponential~\eqref{eq:dyson-series} exists and can be bounded for sufficiently small time intervals by
\begin{equation}\label{eq:U-bound-factorial}\begin{split}
  {\norm{U(\tau,\tau_0) \mathcal{M}}}_{p,v} &\leq \e^{C_0(\tau,\tau_0)} {\norm{\mathcal M}}_{p,w} \\&\quad + \frac{C_0(\tau,\tau_0) K(\upsilon,\omega)^3}{2\uppi} \mleft(\frac\upsilon\omega\mright)^\frac12 \frac{3 - 2C_0(\tau,\tau_0) K(\upsilon,\omega)}{\bigl(1 - C_0(\tau,\tau_0) K(\upsilon,\omega)\bigr)^2} {\norm{\mathcal M}}_{p,w},
\end{split}\end{equation}
where we defined
\begin{equation*}
  K(\upsilon,\omega) \defn \frac{2\upsilon\omega}{\upsilon-\omega}.
\end{equation*}
Note that these inequalities only guarantee the existence of the exponential for a bounded interval of time, even if $\omega=1$ and $\upsilon$ is arbitrarily large.
The existence of such a bound is quite natural in a hyperbolic problem.
Since the radius of convergence is initially bounded, it must shrink as the area of dependence grows with the propagation speed.

Using the result above, we can now show:
\begin{theorem}\label{thm:fact_weights}
  Let $\upsilon > \omega \geq 1$ and $p \geq 1$.
  Suppose that $I = [\tau_0,\tau_1]$ and $V \in C(I)$ such that
  \begin{equation*}
    C_0(\tau_0,\tau_1) K(\upsilon,\omega) < 1.
  \end{equation*}
  The evolution operator $U(\tau,\tau_0)$, defined by~\eqref{eq:dyson-series}, is the unique bounded operator from $\vec\ell^p(w)$ to $\vec\ell^p(v)$ with weights $w_n = (2n)! \omega^{2n}$ and $v_n = (2n)! \upsilon^{2n}$, such that the properties \ref{item:e1}--\ref{item:e3} hold in the interval~$I$ (in the strong sense from $\vec\ell^p(w)$ to $\vec\ell^p(v)$).
  Moreover, for $\mathcal{M} \in \vec\ell^p(w)$, we have the bound~\eqref{eq:U-bound-factorial} and the strong continuity property
  \begin{equation*}
    \lim_{\tau \to \tau_0} {\norm{U(\tau,\tau_0) \mathcal{M} - \mathcal{M}}}_{p,v} = 0.
  \end{equation*}
\end{theorem}
\begin{proof}
  We have already shown that the evolution operator $U(\tau,\tau_0)$ is well-defined as an operator from $\vec\ell^p(w)$ to $\vec\ell^p(v)$.
  In fact, the same estimates show that, for sufficiently small $\varepsilon > 0$, $U(\tau,\tau_0)$ is bounded from $\vec\ell^p(w)$ to $\vec\ell^p(v_\varepsilon)$, where $v_{\varepsilon,n} = (2n)! (\upsilon-\varepsilon)^{2n}$.
  Moreover, applying these estimates again, we find
  \begin{equation*}\begin{split}\MoveEqLeft
    {\norm{U(\tau,\tau_0) \mathcal{M} - \mathcal{M}}}_{p,v} \\&\leq \mleft( \e^{C_0(\tau,\tau_0)}- 1 + \frac{C_0(\tau,\tau_0) K(\upsilon,\omega)^3}{2\uppi} \mleft(\frac\upsilon\omega\mright)^\frac12 \frac{3 - 2 C_0(\tau,\tau_0) K(\upsilon,\omega)}{\bigl(1 - C_0(\tau,\tau_0) K(\upsilon,\omega)\bigr)^2} \mright) {\norm{\mathcal M}}_{p,w} \to 0
  \end{split}\end{equation*}
  as $\tau \to \tau_0$ because $C_0(\tau,\tau_0) \to 0$ in that limit.

  Property~\ref{item:e1} is obvious and~\ref{item:e2} is shown (as for groups generated by bounded operators) by multiplying the series~\eqref{eq:dyson-series} for $U(\tau,r)$ and $U(r,\tau_0)$.
  The final property~\ref{item:e3} is proven by differentiating~\eqref{eq:dyson-series} term by term using the formal relation $\partial_\tau U_n(\tau,\tau_0) = S(\tau) U_{n-1}(\tau,\tau_0)$.
  The resulting expression is well-defined because $U(\tau,\tau_0)$ is bounded from $\vec\ell^p(w)$ to $\vec\ell^p(v_\varepsilon)$ and $S(\tau)$ is bounded from $\vec\ell^p(v_\varepsilon)$ to $\vec\ell^p(v)$.
\end{proof}

Finally, we prove a perturbation theorem:
\begin{theorem}\label{thm:pert-fact_weights}
  In addition to the assumptions of Thm.~\ref{thm:fact_weights}, suppose that $\tilde{V} \in C(I)$ such that\footnote{Note that this condition adds no additional restriction because the role of~$V$ and~$\tilde{V}$ can be exchanged.}
  \begin{equation*}
    C_0(\tau,\tau_0) \geq 2 \abs*{\int_{\tau_0}^\tau \sqrt{1+\tilde{V}(\eta)^2} \dif\eta}
    \quad\text{and}\quad
    2C_0(\tau,\tau_0) K(\upsilon,\omega) < 1.
  \end{equation*}
  Denote by $U(\tau,\tau_0)$ and $\tilde{U}(\tau,\tau_0)$ the evolution operators for $V$ and $\tilde{V}$.
  If $\mathcal{M}, \tilde{\mathcal M} \in \vec\ell^p(w)$, we have
  \begin{equation*}
    {\norm{U(\tau,\tau_0) \mathcal{M} - \tilde{U}(\tau,\tau_0) \tilde{\mathcal M}}}_{p,v} \leq c(\tau) {\norm{\mathcal{M} - \tilde{\mathcal M}}}_{p,w} + 2 c(\tau)^2 {\norm{\tilde{\mathcal M}}}_{p,w} \int_{\tau_0}^\tau \abs{V(\eta)-\tilde{V}(\eta)} \dif\eta,
  \end{equation*}
  where
  \begin{equation*}
    c(\tau) = \e^{C_0(\tau,\tau_0)} + \frac{4 C_0(\tau,\tau_0) K(\upsilon,\omega)^3}{\uppi} \mleft(\frac\upsilon\omega\mright)^{\mathrlap{\frac12}} \frac{3 - 4 C_0(\tau,\tau_0) K(\upsilon,\omega)}{\bigl(1 - 2 C_0(\tau,\tau_0) K(\upsilon,\omega)\bigr)^2}.
  \end{equation*}
\end{theorem}
\begin{proof}
  Set
  \begin{equation*}
    \tilde{w}_n = (2n)! \tilde\omega^{2n},
    \quad
    \tilde\omega = \frac{2\upsilon\omega}{\upsilon+\omega}.
  \end{equation*}
  Observe that
  \begin{equation*}
    \frac{\tilde\omega \omega}{\tilde\omega-\omega} = \frac{2\upsilon\omega}{\upsilon-\omega} = \frac{\upsilon \tilde\omega}{\upsilon-\tilde\omega}
  \end{equation*}
  and
  \begin{equation*}
    \frac{\tilde\omega}{\omega} = \frac{2\upsilon}{\upsilon+\omega} < \frac{\upsilon}{\omega},
    \quad
    \frac{\upsilon}{\tilde\omega} = \frac{\upsilon+\omega}{2\omega} < \frac{\upsilon}{\omega}.
  \end{equation*}
  Further, note that $2C_0(\tau,\tau_0)K(\upsilon,\omega) < 1$ by assumption.
  Therefore,
  \begin{alignat*}{2}
    {\norm{U(\tau,\tau_0) \mathcal{M}}}_{p,v} &\leq c(\tau) {\norm{\mathcal M}}_{p,\tilde{w}}, \quad&
    {\norm{\tilde{U}(\tau,\tau_0) \mathcal{M}}}_{p,v} &\leq c(\tau) {\norm{\mathcal M}}_{p,\tilde{w}}, \\
    {\norm{U(\tau,\tau_0) \mathcal{M}}}_{p,\tilde{w}} &\leq c(\tau) {\norm{\mathcal M}}_{p,w}, \quad&
    {\norm{\tilde{U}(\tau,\tau_0) \mathcal{M}}}_{p,\tilde{w}} &\leq c(\tau) {\norm{\mathcal M}}_{p,w}.
  \end{alignat*}

  Now, note the elementary identity
  \begin{equation*}
    U(\tau,\tau_0) \mathcal{M} - \tilde{U}(\tau,\tau_0) \tilde{\mathcal M} = U(\tau,\tau_0) (\mathcal{M} - \tilde{\mathcal M}) + \bigl(U(\tau,\tau_0) - \tilde{U}(\tau,\tau_0)\bigr) \tilde{\mathcal M}.
  \end{equation*}
  For the first summand we find
  \begin{equation*}
    {\norm{U(\tau,\tau_0) (\mathcal{M} - \tilde{\mathcal M})}}_{p,v}
    \leq {\norm{U(\tau,\tau_0) (\mathcal{M} - \tilde{\mathcal M})}}_{p,\tilde{w}} \leq c(\tau) {\norm{\mathcal{M} - \tilde{\mathcal M}}}_{p,w}.
  \end{equation*}
  For the second summand we find, using the fundamental theorem of calculus (for Banach space-valued integrals),
  \begin{equation*}
    \bigl(U(\tau,\tau_0) - \tilde{U}(\tau,\tau_0)\bigr) \tilde{\mathcal M} = \int_{\tau_0}^\tau U(\tau,\eta) \bigl( \tilde{S}(\eta) - S(\eta) \bigr) \tilde{U}(\eta,\tau_0) \tilde{\mathcal M} \dif\eta,
  \end{equation*}
  and thus
  \begin{align*}
    {\norm{U(\tau,\tau_0) \tilde{\mathcal M} - \tilde{U}(\tau,\tau_0) \tilde{\mathcal M}}}_{p,v}
    &\leq \int_{\tau_0}^\tau \norm[\big]{U(\tau,\eta) \bigl( \tilde{S}(\eta) - S(\eta) \bigr) \tilde{U}(\eta,\tau_0) \tilde{\mathcal M}}_{p,v} \dif\eta \\
    &\leq 2 c(\tau) \int_{\tau_0}^\tau \abs{V(\eta) - \tilde{V}(\eta)} {\norm{\tilde{U}(\eta,\tau_0) \tilde{\mathcal M}}}_{p,\tilde{w}} \dif\eta \\
    &\leq 2 c(\tau)^2 {\norm{\tilde{\mathcal M}}}_{p,w} \int_{\tau_0}^\tau \abs{V(\eta)-\tilde{V}(\eta)} \dif\eta.
  \end{align*}
\end{proof}

As for the case of geometrically growing weights, under more stringent assumptions on the potential $V$, one obtains the following:
\begin{proposition}\label{prop:factorially-smooth}
  If, in addition to the assumptions of Thm.~\ref{thm:fact_weights}, $V \in C^k(I)$, then $\partial_\tau^{k+1} U(\tau,\tau_0)$ is bounded from $\vec\ell^p(w)$ to  $\vec\ell^p(v)$.
\end{proposition}
\begin{proof}
  As in the proof of Thm.~\ref{thm:fact_weights}, $U(\tau,\tau_0)$ is bounded from $\vec\ell^p(w)$ to $\vec\ell^p(v_\varepsilon)$.
  Moreover, it is not difficult to see that products of $S(\tau)$ and its derivatives are bounded from $\vec\ell^p(v_\epsilon)$ to $\vec\ell^p(v)$.
  (Concrete estimates similar to those before Thm.~\ref{thm:fact_weights} could be computed, but because no infinite sums are involved here this is not necessary.)
  The result follows now by iteratively differentiating property \ref{item:e3}.
\end{proof}

As before, this implies that solutions of~\eqref{eq:moment-dynamics} with initial conditions in $\vec\ell^p(w)$ are smooth if $V$ is smooth.

\section{Abstract semiclassical Einstein equation}
\label{sec:abstract-sce}

The SCE on FLRW spacetimes contains only one geometric degree of freedom -- the scale factor $a = a(\tau)$.
Therefore, as described in Sect.~\ref{sec:sce}, it turns into an ODE for the scale factor, coupled to a dynamical system describing the evolution of the state.
In this section, we develop a scheme to solve such systems, encompassing also a large class of modifications of the SCE.
Here we always solve the SCE forward in time but equivalent results hold for solutions backward in time.

For fixed $k \in \NN$ and $\tau \in I \subset \RR$, consider the initial value problem for the quasilinear system
\begin{subequations}\label{eq:quasilinear}\begin{empheq}[left=\empheqlbrace]{align}
  \partial_\tau^{k+1} a(\tau) &= f(\tau, \vec{a}, \mathcal{M}), \label{eq:quasilinear-1} \\
  \partial_\tau \mathcal{M}(\tau) &= S(\tau, \vec{a}) \mathcal{M}(\tau), \label{eq:quasilinear-2}
\end{empheq}\end{subequations}
where
\begin{itemize}
  \item $a=a(\tau) \in \RR$ is the scale factor and $\vec{a} = (a, a', \dotsc, a^{(k)})$ its $k$-jet,
  \item $\mathcal{M}=\mathcal{M}(\tau) \in \vec\ell^p(v)$ is the sequence of coincidence limits, as defined in Sect.~\ref{sec:moments}, for some $p \geq 1$ and sequence of weights~$v$,
  \item $S(\tau, \vec{a}) = S\bigl(\tau, a(\tau), a'(\tau), \dotsc, a^{(k)}(\tau) \bigr)$ is the generator of the dynamics of~$\mathcal{M}$, as defined in~\eqref{eq:S-defn}, for a potential $V(\tau, \vec{a}) = V\bigl(\tau, a(\tau), a'(\tau), \dotsc, a^{(k)}(\tau) \bigr)$ which depends on the $k$-jet of the scale factor and may also have an explicit time-dependence, and
  \item $f(\tau, \vec{a}, \mathcal{M}) = f\bigl(\tau, a(\tau), a'(\tau), \dotsc, a^{(k)}(\tau), \mathcal{M}(\tau)\bigr)$ specifies the dynamics of the scale factor including a possible explicit time-dependence and a back-reaction by the quantum field via~$\mathcal{M}$.
\end{itemize}

We note that the results in this section generalize to states including classical fields, if one simply includes the additional degrees of freedom from Rem.~\ref{rem:background-fields} in the system~\eqref{eq:quasilinear} and adds \eqref{eq:trace-classical} to $f(\tau, \vec{a},\mathcal{M})$.
These modifications do not change the structure of the proofs below.

After developing this abstract formalism in the following two subsections, it will be applied to traced semiclassical Einstein equation in Sect.~\ref{sub:solved-traced-SCE}, by choosing $k = 3$ and a particular function $f = f_\mathrm{tr}$ coming from the trace of the quantized stress-energy tensor.

\subsection{Existence and uniqueness of solutions}

Existence of solutions to~\eqref{eq:quasilinear} can be shown by a (partial) linearization and employing a fixed-point argument via the construction of a contraction map.
With the preparatory results from the previous sections at hand, this theorem is a relatively straightforward adaption of standard results (see e.g.~\cite{kato-quasilinear}) on quasilinear systems to the case of Eq.~\eqref{eq:quasilinear}.
The (standard) proofs are deferred to Sect.~\ref{sec:proofs}.

\begin{theorem}\label{thm:quasilinear}
  Let $\vec{b}_c \in \RR^{k+1}$ and set
  \begin{equation*}
    \mathcal{B}_r \defn \bigl\{ \vec{b} \in \RR^{k+1} \;\big|\; \norm{\vec{b}-\vec{b}_c} \leq r \bigr\},
    \quad
    r > 0,
  \end{equation*}
  where we use the Euclidean norm on $\RR^{k+1}$.
  For fixed $R > 0$ and weight sequence $w$, consider $\vec{a}_\init \in \mathcal{B}_{R/2}$ and $\mathcal{M}_\init \in \vec\ell^p(w)$.
  Suppose that there is $I \defn [\tau_\init,\tau_\fin] \subset \RR$ and a weight sequence~$v$ such that the following holds:
  \begin{enumerate}
    \item
      For each $\vec{a} \in C(I; \mathcal{B}_R)$ there exists a unique solution $\mathcal{M}[\vec a] \in C^1(I; \vec\ell^p(v))$ to~\eqref{eq:quasilinear-2} with initial value $\mathcal{M}[\vec a](\tau_\init) = \mathcal{M}_\init$.
      Set
      \begin{equation*}
        \mu_\mathcal{M} = \sup_{\tau \in I, \vec{a} \in C(I;\mathcal{B}_R)} {\norm{\mathcal{M}[\vec a](\tau)}}_{p,v}.
      \end{equation*}
    \item
      There is $L_{\mathcal M}>0$ such that, for all $\vec{a},\tilde{\vec a} \in C(I; \mathcal{B}_R)$,
      \begin{equation}\label{eq:L_M}
        {\norm{\mathcal{M}[\vec a](\tau)-\mathcal{M}[\tilde{\vec a}](\tau)}}_{p,v} \leq L_{\mathcal M} \int_{\tau_\init}^\tau \norm{\vec{a}(\eta) - \tilde{\vec a}(\eta)} \dif\eta.
      \end{equation}
    \item
      For each $\vec{b},\tilde{\vec{b}} \in \mathcal{B}_R$ and $\mathcal{M},\tilde{\mathcal M} \in \vec\ell^p(v)$ bounded by $\mu_\mathcal{M}$, the map $\tau \mapsto f(\tau,\vec{b},\mathcal{M})$ is continuous for $\tau \in I$, and there are $\mu_f > 0$, $L_f > 0$ such that
      \begin{align}
        \abs{f(\tau,\vec{b},\mathcal{M})} &\leq \mu_f, \label{eq:mu_f} \\
        \abs{f(\tau,\vec{b},\mathcal{M})-f(\tau,\tilde{\vec{b}},\tilde{\mathcal M})} &\leq L_f \bigl(\norm{\vec{b}-\tilde{\vec{b}}} + {\norm{\mathcal{M}-\tilde{\mathcal M}}}_{p,v}\bigr) \label{eq:L_f}
      \end{align}
      uniformly for $\tau \in I$.
  \end{enumerate}
  Then there exists $\tau_\mathrm{stop} \in (\tau_\init,\tau_\fin]$ such that~\eqref{eq:quasilinear} has a unique (local) solution
  \begin{equation*}
    (a, \mathcal{M}) \in C^{k+1}(J) \times C^1(J; \vec\ell^p(v))
  \end{equation*}
  in $J = [\tau_\init,\tau_\mathrm{stop}]$, satisfying the initial conditions $\vec{a}(\tau_\init) = \vec{a}_\init, \mathcal{M}(\tau_\init) = \mathcal{M}_\init$ and the bounds $\vec{a}(\tau) \in \mathcal{B}_R$, $\norm{\mathcal{M}(\tau)}_{p,v} \leq \mu_\mathcal{M}$ for $\tau \in J$.
\end{theorem}

In particular, we can apply the above theorem to the cases studied in the previous section.
In the case of geometrically growing weights we can even obtain maximal or global solutions.

\begin{proposition}\label{prop:quasilinear-special}
  Consider Thm.~\ref{thm:quasilinear} with the weights
  \begin{enumerate}
    \item[\normalfont(a)] $v_n = w_n = \omega^n$ with $\omega > 0$ (geometrically growing weights), or
    \item[\normalfont(b)] $v_n = (2n!) \upsilon^{2n}$ and $w_n = (2n!) \omega^{2n}$ with $\upsilon,\omega > 0$ (factorially growing weights).
  \end{enumerate}
  Further, suppose that $\tau \mapsto V(\tau,\vec{b})$ is continuous for all $\vec{b} \in \mathcal{B}_R$, and there is $L_V > 0$ such that
  \begin{equation*}
    \abs{V(\tau,\vec{b})-V(\tau,\tilde{\vec{b}})} \leq L_V \norm{\vec{b}-\tilde{\vec{b}}}
  \end{equation*}
  for each $\vec{b}, \tilde{\vec{b}} \in \mathcal{B}_R$, uniformly for $\tau$ in some interval~$I'$.
  Then there exists $I = [\tau_\init,\tau_\fin] \subset I'$ such that the assumptions {\normalfont(i)--(iii)} of Thm.~\ref{thm:quasilinear} are satisfied.
\end{proposition}
\begin{proof}
  The assertion follows from Thms.~\ref{thm:geom_weights} and~\ref{thm:pert-geom_weights} for~(a), and from Thms.~\ref{thm:fact_weights} and~\ref{thm:pert-fact_weights} for~(b).
\end{proof}

\begin{proposition}\label{prop:quasilinear-global}
  If the assumptions of Prop.~\ref{prop:quasilinear-special} hold with geometrically growing weights (a), then the unique local solution $(a,\mathcal{M})$ of Thm.~\ref{thm:quasilinear} extends to a maximal solution (the solution exists up to a singularity of $V$ or $f$) or to a global solution (the solution exists for arbitrarily large times).
\end{proposition}
\begin{proof}
  This is shown by gluing local solutions of the initial value problem (using Thm.~\ref{thm:quasilinear} and Prop.~\ref{prop:quasilinear-special}).
\end{proof}

\begin{remark}
  In the case of factorially growing weights, the situation is considerably more complicated because, a priori, already the solution for the dynamics of the moments $\mathcal{M}$ exists only for a bounded time interval.
\end{remark}

In the case of the SCE as described in Sect.~\ref{sec:sce}, $V$ and $f$ (as we will see below) are rational functions in the derivatives of the scale factor~$a$ and $\log(a)$.
Therefore, the assumptions of Thm.~\ref{thm:quasilinear} imply that we can solve~\eqref{eq:quasilinear} with the $k$-jet $\vec{a}$ of $a$ inside a ball~$\mathcal{B}_R$ and thus away from the poles of~$V$ and~$f$.
In the case of geometrically growing weights, we even have maximal (or global) solution.
Moreover, if $V$ and $f$ are smooth (as in our application), also the solution is smooth:

\begin{proposition}\label{prop:quasilinear-smooth}
  Suppose that $(a,\mathcal{M}) \in C^{k+1}(I) \times C^1(I; \vec\ell^p(w))$ is a solution of~\eqref{eq:quasilinear} in the interval $I \subset \RR$ with $\vec{a} \in C^1(I; \mathcal{B}_R)$ (for $\mathcal{B}_R$ as in Thm.~\ref{thm:quasilinear}).
  If $V \in C^l(I \times \mathcal{B}_R)$ and $f \in C^l(I \times \mathcal{B}_R \times \vec\ell^p(w))$ for some $l \in \NN$, then $a \in C^{k+l}(I)$ and $\mathcal{M} \in C^l(I; \ell^p(w))$.
  In particular, if $V$ and $f$ are smooth, the solution is smooth.
\end{proposition}
\begin{proof}
  Follows by iterated derivation and substitution of \eqref{eq:quasilinear} together with Prop.~\ref{prop:factorially-smooth}.
\end{proof}

The regularity of the scale factor is especially relevant when discussing the regularity of the state.
In fact, the state can only be Hadamard if the scale factor is smooth.
In that case, if the state is initially Hadamard, it will be Hadamard for all time, because the Hadamard property is propagated on smooth spacetimes \cite{radzikowski-verch}.

\subsection{Continuous dependence on initial data and parameters}

To study the continuous dependence of the abstract SCE~\eqref{eq:quasilinear} on its initial data, the potential $V$ and the `back-reaction' function $f$, we consider another quasilinear equation of the form~\eqref{eq:quasilinear}:
\begin{subequations}\label{eq:quasilinear-tilde}\begin{empheq}[left=\empheqlbrace]{align}
  \partial_\tau^{k+1} \tilde{a}(\tau) &= \tilde{f}(\tau, \tilde{\vec a}, \tilde{\mathcal M}), \\
  \partial_\tau \tilde{\mathcal M}(\tau) &= \tilde{S}(\tau, \tilde{\vec a}) \tilde{\mathcal M}(\tau),
\end{empheq}\end{subequations}
with initial conditions $\tilde{\vec a}_\init \in \RR^{k+1}$ and $\tilde{\mathcal M}_\init \in \vec\ell^p(w)$. Above, $\tilde{S}(\tau, \tilde{\vec a})$ is given by~\eqref{eq:S-defn} for a potential $\tilde{V}(\tau, \tilde{\vec a})$.

\begin{theorem}\label{thm:continuous}
  Suppose that the coefficients of~\eqref{eq:quasilinear} and~\eqref{eq:quasilinear-tilde} satisfy the assumptions {\normalfont (i), (iii)} of Thm.~\ref{thm:quasilinear} for the same constants.
  Furthermore, (with the obvious notation) assume that, instead of assumption {\normalfont (ii)} of Thm.~\ref{thm:quasilinear}, we have:
  \begin{enumerate}
    \item[(ii')]
      There is $L_{\mathcal M}' > 0$ such that, for all $\vec{a}, \tilde{\vec a} \in C(I; \mathcal{B}_R)$,
      \begin{equation*}
        {\norm{\mathcal{M}[\vec a](\tau) - \tilde{\mathcal M}[\tilde{\vec a}](\tau)}}_{p,v} \leq L_{\mathcal M}' \mleft( {\norm{\mathcal{M}_\init - \tilde{\mathcal M}_\init}}_{p,w} + \int_{\tau_\init}^\tau \norm{\vec{a}(\eta) - \tilde{\vec a}(\eta)} \dif\eta \mright).
      \end{equation*}
  \end{enumerate}
  Then the two solutions $(a, \mathcal{M})$ and $(\tilde{a}, \tilde{\mathcal M})$ on $J = [\tau_\init,\tau_\mathrm{stop}] \subset I$ satisfy
  \begin{equation}\label{eq:dependence-ineq}
    \norm{\vec{a}(\tau) - \tilde{\vec a}(\tau)} \leq \norm{\vec{a}_\init - \tilde{\vec a}_\init} \e^{K(\tau,\tau_\init)} + (\mu_f' + L_f L_\mathcal{M}' \norm{\mathcal{M}_\init - \tilde{\mathcal M}_\init}_{p,w}) \int_{\tau_\init}^\tau \e^{K(\tau,\eta)} \dif\eta
  \end{equation}
  for $\tau \in J$, where we set $K(\tau,\eta) = (1+L_f) (\tau - \eta) + \frac12 L_f L_\mathcal{M}' (\tau - \eta)^2$ and
  \begin{equation*}
    \mu_f' = \sup{} \abs{f(\tau, \vec{b}, \mathcal{M}) - \tilde{f}(\tau, \vec{b}, \mathcal{M})}
  \end{equation*}
  with the supremum being taken over over $\tau \in J$, $\vec{b} \in \mathcal{B}_R$ and $\norm{\mathcal M}_{p,v} \leq \mu_\mathcal{M}$.
\end{theorem}
The proof can be found in Sect.~\ref{sec:proofs}.

The continuous dependence of the solution on the initial data also implies that the effect of any error in the initial data is bounded, with an error bound that increases in time.
This is of particular importance in this approach, as the state represented via its moments essentially requires an infinite amount of initial data (even if this can in principle be encoded in a single function).
The error caused by a truncation of the sequence of moments at some order can thus be controlled.

\subsection{Application to the traced SCE}
\label{sub:solved-traced-SCE}

In this subsection we show how the results above may be applied to the SCE as discussed in Sect.~\ref{sec:sce}.

Recall that $\omega_2^\mathrm{reg} = \omega_2 - H_n$, where $n$ is sufficiently large.
It follows by a straightforward calculation from the definitions of Sect.~\ref{sec:scalar} that
\begin{align*}
  [\omega_2^\mathrm{reg}] &= \frac{1}{a^2} \bigl( \mathcal{M}_{\varphi\varphi,0} + \mathcal{H}_{\varphi\varphi,1} - \tilde{\mathcal{H}}_{\varphi\varphi,1} \bigr) \Bigl|_{r=0}, \\
  [(\one \otimes \upDelta) \omega_2^\mathrm{reg}] &= \frac{1}{a^2} \bigl( \mathcal{M}_{\varphi\varphi,1} + \upDelta_r (\mathcal{H}_{\varphi\varphi,2} - \tilde{\mathcal{H}}_{\varphi\varphi,2}) \bigr) \Bigl|_{r=0}, \\
  \begin{split}
    [(\partial_\tau \otimes \partial_\tau) \omega_2^\mathrm{reg}] &= \frac{1}{a^2} \bigl( \mathcal{M}_{\pi\pi,0} + \mathcal{H}_{\pi\pi,2} - \tilde{\mathcal{H}}_{\pi\pi,2} \bigr) + \frac{a^{\prime\,2}}{a^4} \bigl( \mathcal{M}_{\varphi\varphi,0} + \mathcal{H}_{\varphi\varphi,1} - \tilde{\mathcal{H}}_{\varphi\varphi,1} \bigr) \\&\quad - 2\frac{a'}{a^3} \bigl( \mathcal{M}_{(\varphi\pi),0} + \mathcal{H}_{(\varphi\pi),2} - \tilde{\mathcal{H}}_{(\varphi\pi),2} \bigr) \Bigl|_{r=0}.
  \end{split}
\end{align*}

Thus, combining the results from Sects.~\ref{sec:scalar} and~\ref{sec:sce}, the traced SCE~\eqref{eq:traced-SCE} can be expanded to the rather long equation
\begin{equation*}\begin{split}
  0 &=
  \mleft( -12 (3 c_3 + c_4) -\frac{1}{480\uppi^2} + \frac{6\xi-1}{48\uppi^2} + \frac{(6\xi-1)^2}{16\uppi^2} \log(a \lambda_0) \mright) \\&\quad \times \mleft( \frac{a^{(4)}}{a^5} - 4 \frac{a^{(3)} a'}{a^6} - 3 \frac{a^{\prime\prime\,2}}{a^6} + 6 \frac{a'' a^{\prime\,2}}{a^7} \mright) \\&\quad + \frac{(6\xi-1)^2}{32\uppi^2} \mleft( 4 \frac{a^{(3)} a'}{a^6} + 3 \frac{a^{\prime\prime\,2}}{a^6} - 10 \frac{a'' a^{\prime\,2}}{a^7} \mright) + \frac{1}{240\uppi^2} \mleft( -\frac{a'' a^{\prime\,2}}{a^7} + \frac{a^{\prime\,4}}{a^8} \mright) \\&\quad + \mleft( \frac{6}{\kappa} + m^2 \Bigl( -6 c_2 + \frac{1}{48\uppi^2} + \frac{6\xi-1}{8\uppi^2} \bigl(1 + \log(a \lambda_0) \bigr) \Bigr) \mright) \frac{a''}{a^3} \\&\quad + \frac{(6\xi-1) m^2}{16\uppi^2} \frac{a^{\prime\,2}}{a^4} + m^4 \mleft( 4 c_1 + \frac{1}{32\uppi^2} + \frac{1}{8\uppi^2} \log(a \lambda_0) \mright) - \frac{m^2}{a^2} \mathcal{M}_{\varphi\varphi,0} \\&\quad + (6\xi-1) \mleft( \Bigl( 6\xi \frac{a''}{a^5} - \frac{a^{\prime\,2}}{a^6} + \frac{m^2}{a^2} \Bigr) \mathcal{M}_{\varphi\varphi,0} + 2 \frac{a'}{a^5} \mathcal{M}_{(\varphi\pi),0} - \frac{1}{a^4} \bigl( \mathcal{M}_{\pi\pi,0} + \mathcal{M}_{\varphi\varphi,1} \bigr) \mright).
\end{split}\end{equation*}
We remark that the first line of this equation is due to terms proportional to $\Box R$.

In the general case, this equation can be rewritten as quasilinear fourth order equation of the form
\begin{equation}\label{eq:SCE-ODE}
  \partial_\tau a^{(3)} = f_\mathrm{tr}(a, a', a'', a^{(3)}, \mathcal{M}_0, \mathcal{M}_1),
\end{equation}
which can be solved using Thm.~\ref{thm:quasilinear} (see also Props.~\ref{prop:quasilinear-special} and~\ref{prop:quasilinear-global}).
Note that the right-hand side has poles at $a=0$ and for $a$ such that
\begin{equation}\label{eq:log-singularity}
  30 (6\xi-1)^2 \log(a \lambda_0) = 11 + 5760\uppi^2 (3c_3+c_4) - 60\xi,
\end{equation}
which must be taken into account for the choice of the ball $\mathcal{B}_R$ in Thm.~\ref{thm:quasilinear}.

We do not see the instability near the Minkowski solution described in~\cite{suen:stability1,suen:stability2}, viz., the equation is continuous in a neighbourhood of the Minkowski solution.
The reason is, clearly, that we based our analysis on the fourth order equation given by the traced SCE.
Indeed, inspection of \eqref{eq:T-from-T_00} immediately reveals, that the third order equation given by the energy component of the stress-energy tensor can only be used if $a'$ is bounded away from zero and in that case it is sufficient to work with the third order equation, see also the next subsection.

However, in the non-conformally coupled case the singularity~\eqref{eq:log-singularity} appears.
It can be seen that the position of this singularity depends on the choice of $\lambda_0$, relating the length scales~$\mu$ and~$\lambda$ in~\eqref{eq:mu-lambda-convention}.
While the choice of $\lambda$ is related to the value of the renormalization parameters $c_1, c_2, c_3$ and $c_4$, the value of $\mu$ (and thus $\lambda_0$) is  related to our construction of $\mathcal{M}$.
Of course, changing $\lambda_0$ is not without effect and requires a corresponding change of $\mathcal{M}$.
It is important to note that such a modification of $\mathcal{M}$ will generally cause it to fall out of the sequence space.

There is only one special case in which this equation reduces to a lower than fourth order equation: if
\begin{equation}\label{eq:2nd-order-conditions}
  \xi = \frac16
  \quad\text{and}\quad
  3c_3+c_4 = -\frac{1}{5760\uppi^2},
\end{equation}
the fourth and third order terms drop out and the equation can be rewritten as
\begin{equation*}\begin{split}
   a'' &= \mleft( \frac{a^{\prime\,2}}{a^4} - 1440\uppi^2 \kappa^{-1} + (1440\uppi^2 c_2 - 5) m^2 \mright)^{-1} \\&\quad \times \mleft( \frac{a^{\prime\,4}}{a^5} + \frac12 m^4 a^3 \bigl( 1920\uppi^2 c_1 + 15 + 60 \log(a \lambda_0) \bigr) - 240\uppi^2 m^2 a \mathcal{M}_{\varphi\varphi,0} \mright)
\end{split}\end{equation*}
and thus has a the correct form to be solved by the methods above.
This equation is equivalent to that already considered in~\cite{pinamonti-siemssen}.
Note that the right-hand side has poles at $a=0$ (`big bang/crunch') and
\begin{equation}\label{eq:2nd-order-singularity}
  \frac{a^{\prime\,2}}{a^4} = 1440\uppi^2 \kappa^{-1} - (1440\uppi^2 c_2 - 5) m^2,
\end{equation}
viz., for a certain value of the square of the Hubble parameter ($a^{-2} a'$ in conformal time), which must be taken into account for the choice of the ball~$\mathcal{B}_R$ in Thm.~\ref{thm:quasilinear}.
Also the instability~\eqref{eq:2nd-order-singularity} can be considered irrelevant because $c_2$ should be set to zero as it corresponds to a renormalization of Newton's gravitational constant, which has already been measured, and thus, if the scalar field mass is much smaller than the Planck mass, the singularity occurs for a Hubble parameter close to the inverse Planck time (very many orders of magnitude larger than the currently observed value).

We sum up the results of this subsection:
\begin{theorem}\label{thm:traced_sce}
  The traced SCE is of the form~\eqref{eq:quasilinear} and can be locally solved (using Thm.~\ref{eq:quasilinear} and Props.~\ref{prop:quasilinear-special}, \ref{prop:quasilinear-global}, \ref{prop:quasilinear-smooth}), yielding a unique smooth solution which depends continuously on the initial data and parameters.
  In the case of geometrically growing weights for the moments $\mathcal{M}$ (i.e., vacuum-like states), a maximal or global solution exists.
  In the generic fourth order case, the maximal solution exists up to a big bang/crunch $a = 0$, the logarithmic singularity~\eqref{eq:log-singularity}, or a blow-up $a \to \infty$ (big rip).
  In the second order case (i.e., \eqref{eq:2nd-order-conditions} are satisfied), the maximal solutions exists up to a big bang/crunch, the singularity~\eqref{eq:2nd-order-singularity}, or a big rip.
\end{theorem}

Let us also briefly explain the philosophy behind the choice of the renormalization parameters $c_\bullet$.
The approach described above yields a family of equations for each value of these parameters.
The solutions for each parameter can then, in principle, be compared with experimental/observational data to fix these parameters.
In this context it should be stressed that the same parameters must be chosen for all spacetimes and all experiments in order for the renormalization procedure to satisfy various standard assumptions, including local covariance \cite{hollands-wald:wick,hollands-wald:time-order,hollands-wald:stress-energy}.

\subsection{Application to the energy component of the SCE}
\label{sub:sce-energy}

Analogous to the traced SCE, the energy component of the SCE can be expanded as
\begin{equation}\label{eq:SCE-T00}\begin{split}
  0 &= \mleft( 6 (3 c_3 + c_4) + \frac{1}{960\uppi^2} - \frac{6\xi-1}{96\uppi^2} - \frac{(6\xi-1)^2}{32\uppi^2} \log(a \lambda_0) \mright) \mleft( 2 \frac{a^{(3)} a'}{a^4} - \frac{a^{\prime\prime\,2}}{a^4} - 4 \frac{a'' a^{\prime\,2}}{a^5} \mright) \\&\quad - \frac{(6\xi-1)^2}{16\uppi^2} \frac{a'' a^{\prime\,2}}{a^5} + \frac{1}{960\uppi^2} \frac{a^{\prime\,4}}{a^6} + \mleft( -\frac{3}{\kappa} + m^2 \Bigl( 3 c_2 - \frac{1}{96\uppi^2} - \frac{6\xi-1}{16\uppi^2} \bigl(1 + \log(a \lambda_0) \bigr) \Bigr) \mright) \frac{a^{\prime\,2}}{a^2} \\&\quad - m^4 \mleft( c_1 + \frac{1}{32\uppi^2} \log(a \lambda_0) \mright) a^2 + \frac{m^2}{2} \mathcal{M}_{\varphi\varphi,0} + (6\xi-1) \mleft( -\frac{a^{\prime\,2}}{2a^4} \mathcal{M}_{\varphi\varphi,0} + \frac{a'}{a^3} \mathcal{M}_{(\varphi\pi),0} \mright) \\&\quad + \frac{1}{2a^2} \bigl( \mathcal{M}_{\pi\pi,0} - \mathcal{M}_{\varphi\varphi,1} \bigr),
\end{split}\end{equation}
In the special case~\eqref{eq:2nd-order-conditions}, already considered above, this equation simplifies to
\begin{equation*}\begin{split}
  0 &= \frac{1}{960\uppi^2} \frac{a^{\prime\,4}}{a^6} + \mleft( -\frac{3}{\kappa} + m^2 \Bigl( 3 c_2 - \frac{1}{96\uppi^2} \Bigr) \mright) \frac{a^{\prime\,2}}{a^2} - m^4 \mleft( c_1 + \frac{1}{32\uppi^2} \log(a \lambda_0) \mright) a^2 \\&\quad + \frac{m^2}{2} \mathcal{M}_{\varphi\varphi,0} + \frac{1}{2a^2} \bigl( \mathcal{M}_{\pi\pi,0} - \mathcal{M}_{\varphi\varphi,1} \bigr).
\end{split}\end{equation*}

The general equation can be rewritten as a quasilinear third order equation of the form
\begin{equation*}
  \partial_\tau a'' = f_\mathrm{en}(a, a', a'', \mathcal{M}_0, \mathcal{M}_1),
\end{equation*}
i.e., degree lower than the corresponding equation for the trace.
Just like the analogous equation for the traced SCE, this equation can be plugged into the abstract machinery developed in the first half of this section.
However this comes at the price that $f_\mathrm{en}$ has a pole at $a' = 0$ and thus it cannot be applied there.

Recall that the energy component of the SCE must be implemented as a constraint on the initial data.
An approach to constructing initial data that satisfies the energy constraint and also the Hadamard condition is discussed in Sect.~\ref{sub:initial-data}.

\subsection{Minkowski solution}

Before we continue to discuss the general case again and study the problem of constructing physically reasonable initial data, let us consider the simplest solution of the equations above: the Minkowski solution.

For proper physical initial data it should be required that the initial data for the state (resp. the moments) satisfy the energy (constraint) equation given by~\eqref{eq:T00-FLRW} or, equivalently, by~\eqref{eq:SCE-T00}.
Therefore, on Minkowski spacetime, the following relation needs to hold:
\begin{align*}
  0
  &= \frac12 \mathcal{M}_{\pi\pi,0} - \frac12 \mathcal{M}_{\varphi\varphi,1} + \frac{m^2}{2} \mathcal{M}_{\varphi\varphi,0} - m^4 \mleft( c_1 + \frac{1}{32\uppi^2} \log(\lambda_0) \mright) \\
  &= \mathcal{M}_{\pi\pi,0} - m^4 \mleft( c_1 + \frac{1}{32\uppi^2} \log \lambda_0 \mright),
\end{align*}
where for the last equality we assumed a time-translation invariant state (so that the left-hand side of~\eqref{eq:moment-dynamics} vanishes).
A curious observation is that the energy constraint equation fixes the renormalization parameter~$c_1$.

For example, using the results from Sects.~\ref{sec:vacuum} and~\ref{sec:thermal}, we find
\begin{equation*}
  \mathcal{M}_{\pi\pi,0} = \frac{m^4}{32\uppi^2} \mleft( \frac14 + \log\Bigl( \frac12 m \mu \Bigr) \mright) = \frac{m^4}{32\uppi^2} \mleft( -\frac34 + \gamma + \frac12 \log\Bigl( \frac12 m^2 \lambda^2 \Bigr) + \log \lambda_0 \mright),
\end{equation*}
for the Minkowski vacuum and
\begin{equation*}
  \mathcal{M}_{\pi\pi,0} = \frac{\uppi^4}{15 \beta^4}
\end{equation*}
for the thermal state.
That is, in these examples, $c_1$ is fixed once the mass $m$ and the regularization scale $\lambda$, resp.\ the inverse temperature $\beta$ are fixed.
Note that the arbitrary scale $\lambda_0$ cancels out.

It follows then from~\eqref{eq:R-from-G00} and~\eqref{eq:T-from-T_00}, for any time-translation invariant state on Minkowski spacetime, that
\begin{equation*}
  G_{00} = \kappa {\langle T_{00}^\mathrm{ren} \rangle}_\omega
  \quad\Longleftrightarrow\quad
  -R = \kappa {\langle T^\mathrm{ren} \rangle}_\omega.
\end{equation*}
Therefore we have:
\begin{theorem}
  Any time-translation invariant state on Minkowski spacetime satisfies the SCE on Minkowski spacetime, provided the renormalization parameter~$c_1$ is chosen consistently with the energy constraint.
\end{theorem}

\subsection{Construction of non-trivial physical initial data}
\label{sub:initial-data}

It is not clear how to give initial values for a (Hadamard) state unless the scale function is known in a neighbourhood of the Cauchy surfaces.
This is another reason for why it is not clear how to pose a satisfying initial value problem for the SCE.
Our use of the moments $\mathcal{M}$ instead of a state does not completely solve this problem as it is not clear which sequences of moments belong to \emph{positive} two-point functions of \emph{physical} states.

To escape this problem, we propose the following approach:
We consider a (fixed) FLRW spacetime on which we know how to construct a Hadamard state, e.g., a spacetime which is Minkowskian in some time interval (in the context of Lem.~\ref{lem:tow-in} this would be a neighbourhood of $\tau_\tow$), and then calculate the corresponding moments.
Then, at some point in time (denoted $\tau_\init$ below), we start to switch on the SCE.
If this is done smoothly, the resulting solution will be smooth and the propagated state will remain Hadamard.
Hence, at any point after the SCE is fully switched on (denoted $\tau_\free$ below), we have a solution to the SCE with a Hadamard state.

\begin{figure}
  \centering
  \includegraphics{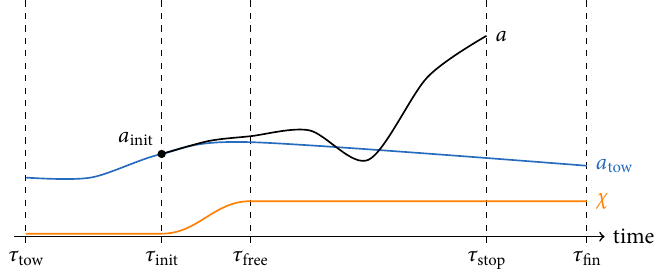}
  \caption{Schematic illustration of the tow-in procedure described in Lem.~\ref{lem:tow-in}, including the various relevant time intervals. In the past of $\tau_\init$ the SCE is `turned off' and the solution is determined by a chosen scale factor $a_\tow$ which evolves to the desired initial value $a_\init$. Then, in the interval $(\tau_\init, \tau_\free)$, the SCE is smoothly turned on using the function $\chi$. Shrinking this time interval, the solution at $\tau_\free$ can be brought arbitrarily close to its value at $\tau_\init$. Finally, after $\tau_\free$ the solution evolves freely via the SCE without the forcing caused by $a_\tow$. \label{fig}}
\end{figure}

To put this idea on a solid footing, we prove the following technical lemma, some aspects of which are also illustrated in Fig.~\ref{fig}:
\begin{lemma}\label{lem:tow-in}
  Let $[\tau_\tow, \tau_\fin] \subset \RR$, $\tau_\init \in (\tau_\tow, \tau_\fin)$ and $a_\tow \in C^{k+1}[\tau_\tow, \tau_\fin]$.
  Suppose that $\mathcal{M}_\tow \in C^1([\tau_\tow, \tau_\fin]; \vec\ell^p(w))$ is a solution of \eqref{eq:quasilinear-2} for the scale factor $a_\tow$ and the weight sequence $w$.
  Further, given appropriately chosen $\vec{b}_c$, $R$, $v$, suppose that the assumptions of Thm.~\ref{thm:quasilinear} hold for the initial data $\vec{a}_\init = \vec{a}_\tow(\tau_\init)$, $\mathcal{M}_\init = \mathcal{M}_\tow(\tau_\init)$.
  Then there exists $\tau_\mathrm{stop} \in (\tau_\init, \tau_\fin]$ such that for any $\varepsilon > 0$ one can find $\tau_\free \in (\tau_\init, \tau_\mathrm{stop})$ and $(a, \mathcal{M})$ satisfying
  \begin{enumerate}
    \item $(a, \mathcal{M}) = (a_\tow, \mathcal{M}_\tow)$ in $[\tau_\tow,\tau_\init]$,
    \item $(a, \mathcal{M})$ solves \eqref{eq:quasilinear} in $[\tau_\free, \tau_\mathrm{stop}]$,
    \item $(a, \mathcal{M}) \in C^{k+1}(J') \times C^1(J';\vec\ell^p(v))$ with $J' = [\tau_\tow, \tau_\mathrm{stop}]$,
    \item $\norm{\vec{a}(\tau_\free) - \vec{a}(\tau_\init)} \leq \varepsilon$ and $\norm{\mathcal{M}(\tau_\free) - \mathcal{M}(\tau_\init)}_{p,v} \leq \varepsilon$,
    \item $(a, \mathcal{M})$ is smooth if $a_\tow, f$ and $V$ are smooth.
  \end{enumerate}
\end{lemma}
\begin{proof}
  For any $\chi \in C^\infty([\tau_\tow, \tau_\fin]; [0,1])$ define
  \begin{equation*}
    f_\chi(\tau, \vec{a}, \mathcal{M}) = \bigl( 1-\chi(\tau) \bigr) a_\tow^{(k+1)}(\tau) + \chi(\tau) f(\tau, \vec{a}, \mathcal{M}).
  \end{equation*}
  We want to apply Thm.~\ref{eq:quasilinear} with $f_\chi$ instead of $f$.
  Assumptions (i) and (ii) of Thm.~\ref{eq:quasilinear} continue to hold with the same constants.
  Estimating $\chi$ by~$1$, \eqref{eq:mu_f} in assumption (iii) can be replaced by
  \begin{equation}\label{eq:mu_f_chi}
    \abs{f_\chi(\tau, \vec{a}, \mathcal{M})} \leq \sup_{\eta \in I}{} \abs{a_\tow^{(k+1)}(\eta)} + \mu_f,
  \end{equation}
  which holds for $\tau \in I$, and the right-hand side of \eqref{eq:L_f} is unchanged.
  Hence Thm.~\ref{thm:quasilinear} can be applied to find $\tau_\mathrm{stop} \in (\tau_\init, \tau_\fin]$ and a solution
  \begin{equation*}
    (a_\chi, \mathcal{M}_\chi) \in C^{k+1}(J) \times C^1(J; \vec\ell^p(v))
  \end{equation*}
  in $J = [\tau_\init, \tau_\mathrm{stop}]$ for \eqref{eq:quasilinear} with $f$ replaced by $f_\chi$.

  Define $(a, \mathcal{M})$ by gluing $(a_\tow, \mathcal{M}_\tow)$ and $(a_\chi, \mathcal{M}_\chi)$.
  Choosing $\chi$ such that $\chi(\tau) = 0$ for $\tau \leq \tau_\init$ and $\chi(\tau) = 1$ for $\tau \geq \tau_\free$ with $\tau_\free \in (\tau_\init, \tau_\mathrm{stop})$ arbitrary, properties (i)--(iii) are obviously satisfied.
  Since estimates in the previous paragraph are uniform in $\chi$, (iv) follows by continuity of the solution, viz., if $\tau_\free$ is chosen sufficiently close to $\tau_\init$, the estimates hold.
  Finally, in the smooth case, apply Prop.~\ref{prop:quasilinear-smooth} to see that (v) holds.
\end{proof}

Applied to the traced SCE, where $V$ and $f$ are smooth for appropriately chosen $\mathcal{B}_R$ such that $a$ is bounded away from zero, this lemma together with the propagation of the Hadamard property on smooth spacetimes \cite{radzikowski-verch} shows:
\begin{proposition}
  If $\mathcal{M}_\tow$ is given by a Hadamard state on the FLRW spacetime with smooth scale function $a_\tow$, the solution $(a, \mathcal{M})$ in Lem.~\ref{lem:tow-in} has moments~$\mathcal{M}$ for a Hadamard state with respect to~$a$.
\end{proposition}

While the `tow in' construction described above can be used to prove the existence of Hadamard states fulfilling the trace equation of the SCE, the status of the energy constraint remains unsatisfactory:
Although the SCE preserves the energy constraint, this is not case for the deformed version used in the proof of Lem.~\ref{lem:tow-in}.
In other words, although the solutions constructed by that lemma satisfy the traced SCE, they do not satisfy the energy equation and hence not the full SCE.
To fix this issue, we need to amend this construction through an additional step.
Namely, we give up control over the initial datum of the third derivative of the scale factor, and change it in such a way that the energy constraint is satisfied again from $\tau_\free$ onwards.
That this can indeed be done consistently is shown in the following Theorem.

\begin{theorem}\label{thm:initial_data}
  Given $b_0, b_1, b_2 \in \RR$ such that
  \begin{equation}\label{eq:b-conditions}
    6 (3 c_3 + c_4) + \frac{1}{960\uppi^2} - \frac{6\xi-1}{96\uppi^2} - \frac{(6\xi-1)^2}{32\uppi^2} \log(b_0 \lambda_0) \neq 0,
    \quad
    b_0 > 0,
    \quad
    b_1 \neq 0,
  \end{equation}
  and $\varepsilon > 0$, there exists a time interval $J = [\tau_\free, \tau_\mathrm{stop}] \subset \RR$, a weight sequence $w$ and $(a, \mathcal{M}) \in C^\infty(J) \times C^\infty(J; \vec\ell^p(w))$ such that
  \begin{enumerate}
    \item there is a Hadamard state $\omega$ with associated moments $\mathcal{M}$ for the scale factor $a$,
    \item $(a, \omega)$ satisfies the SCE, viz., both trace equation \eqref{eq:traced-SCE} and the energy constraint \eqref{eq:constraint-SCE},
    \item $\norm[\big]{a^{(j)}(\tau_\free) - b_j} \leq \varepsilon$ for $j = 0,1,2$.
  \end{enumerate}
\end{theorem}
\begin{proof}
  Let $\chi \in C^\infty(\RR; [0,1])$ with $\supp\chi = [-1, 1]$ and $\chi(0) = 1$.
  For $\delta > 0$, define $\chi_\delta(x) = \chi(\frac{x}{\delta})$.

  Fix $[\tau_\tow, \tau_\fin] \subset \RR$ and $\tau_\init \in (\tau_\tow, \tau_\fin)$.
  Let $a_\tow \in C^\infty[\tau_\tow, \tau_\fin]$ with $a_\tow > 0$ be initially Minkowskian (i.e., $a_\tow = 1$ in a neighbourhood of $\tau_\tow$) and such that
  \begin{equation*}
    a_\tow^{(j)}(\tau_\init) = b_j \quad \text{for $j=0,1,2$}.
  \end{equation*}
  Then the functions
  \begin{equation*}
    a_{\tow,c,\delta}(\tau) = a_\tow(\tau) + \frac{c}{2} \int_{\tau_\tow}^\tau (\tau - \eta)^2 \chi_\delta(\eta-\tau_\init) \dif\eta,
  \end{equation*}
  defined for $c \in \RR$ and $\delta > 0$, satisfy
  \begin{align}
    \abs[\big]{a_{\tow,c,\delta}^{(j)}(\tau_\init) - b_j} &\leq c \delta^{3-j}
    \quad
    \text{for $j=0,1,2$}, \label{eq:a_tow-c-delta-estimate}
    \\
    a_{\tow,c,\delta}'''(\tau_\init) &= a_\tow'''(\tau_\init) + c. \notag
  \end{align}
  In particular, given any $c \in \RR$, if $\delta > 0$ is sufficiently small, the conditions \eqref{eq:b-conditions} hold with $b_0, b_1$ replaced by $a_{\tow,c,\delta}(\tau_\init)$ and $a_{\tow,c,\delta}'(\tau_\init)$.

  Let $\omega_\tow$ be a translation-invariant Hadamard state on Minkowski spacetime such that the associated moments are (at most) geometrically growing (e.g., choose the vacuum state).
  Denote $(a_{c,\delta}, \mathcal{M}_{c,\delta})$ the solution obtained from Lem.~\ref{lem:tow-in} with $a_{\tow,c,\delta}$ and initial moments $\mathcal{M}_\tow$ given by $\omega_\tow$.
  Observe that $\tau_\free$ and $\tau_\mathrm{stop}$ can be chosen uniformly for $c$ from any fixed bounded interval.

  Written as
  \begin{equation*}
    g(\vec{a}, \mathcal{M}, \tau) = a'''(\tau) - f_\mathrm{en}(a, a', a'', \mathcal{M}) \bigr|_\tau
  \end{equation*}
  with the right-hand side as defined in Sect.~\ref{sub:sce-energy}, it is evident that the energy constraint is continuous in $\vec{a}$ for $a > 0$ and $a' \neq 0$.
  Since $V$ and $f_\mathrm{en}$ only depend on lower than third order derivatives of the scale factor, it then follows from Thm.~\ref{thm:pert-geom_weights} and \eqref{eq:a_tow-c-delta-estimate} that for any $\gamma > 0$ there are $c_\pm \in \RR$ and $\delta > 0$ such that
  \begin{equation*}
    \pm g(\vec{a}_{c_\pm,\delta}, \mathcal{M}_{c_\pm,\delta}, \tau_\init) \geq \gamma
  \end{equation*}
  It then follows by (iv) of Lem.~\ref{lem:tow-in}, that one can choose $\tau_\free$ such that $\pm g(\vec{a}_{c_\pm,\delta}, \mathcal{M}_{c_\pm,\delta}, \tau_\free) > 0$.
  Moreover, given sufficiently small $\delta > 0$, Thm.~\ref{thm:continuous} implies the continuous dependence of $g(\vec{a}_{c,\delta}, \mathcal{M}_{c,\delta}, \tau_\free)$ on $c$.
  An application of the intermediate value theorem thus shows that there exists $c_0 \in (c_-,c_+)$ such that the associated solution $(a, \mathcal{M}) = (a_{c_0\delta}, \mathcal{M}_{c_0,\delta})$ satisfies the energy constraint $g(\vec{a}, \mathcal{M}) = 0$ at $\tau_\free$.

  As the energy constraint is propagated by solutions to the SCE, $(a, \mathcal{M})$ is a proper solution to the SCE (i.e., both the energy and the trace equation) in $[\tau_\free, \tau_\mathrm{stop}]$.
  As the Hadamard property is propagated on smooth spacetimes, $\mathcal{M}$ is associated to a Hadamard state, namely the state $\omega_\tow$ propagated on $a$.
  This shows (i) and (ii).
  Finally, (iii) follows from the triangle inequality using \eqref{eq:a_tow-c-delta-estimate} and (iv) of Lem.~\ref{lem:tow-in}.
\end{proof}

This theorem shows, in particular, that the set of solutions satisfying the Hadamard condition to the SCE is non-empty.
In fact, for any choice of approximate initial conditions for the scale factor, a nearby solution of the SCE exists.

\subsection{Reconstruction of the quantum state}

There is no obvious way to directly relate a sequence of moments~$\mathcal{M}$ to a quasi-free state.
Note that this is not a classical moment problem, as the degree $l$ of the counter terms in~\eqref{eq:moments} is neither fixed nor unique.
Here we can easily circumvent this problem:
Suppose that in addition to the initial data for the moments~$\mathcal{M}$ we are given initial data for the associated state, or rather the two-point function.
(Concretely, this is possible, for instance, in the approach described in Sect.~\ref{sub:initial-data}.)
Then, once we have solved the quasilinear system~\eqref{eq:quasilinear} for the scale factor and the moments, we can evolve the initial data for the two-point function with the obtained scale factor.
Necessarily, the evolved two-point function is compatible with the evolved moments.
Alternatively, we can augment~\eqref{eq:quasilinear} by adding a third equation for the two-point function and directly co-evolve it with the moments.

\section{Outlook}
\label{sec:outlook}

Here we restricted ourselves to the SCE with a free scalar field on flat cosmological spacetimes.
An extension of our approach to other non-interacting types of matter (e.g., the Dirac field) and non-flat cosmological spacetime seems feasible.
The former requires some modifications of the moment spaces, while the latter requires a careful treatment of homogeneous distributions on maximally symmetric spaces.
Currently we are also investigating the specialization of some of our methods to (cosmological) de Sitter spacetime \cite{nicolai:desitter}, where we expect to extend some of the results of \cite{juarez=aubry:desitter}.

The introduction of potential energy of the field (neglecting self-interaction) $-\lambda \langle{:}\phi^4{:}\rangle_\omega$ would lead to quadratic terms in $\mathcal{M}_{\varphi\varphi}$ and further correction terms in the trace equation that nevertheless still fall in the class of the abstract SCE~\eqref{eq:quasilinear}.
Such modifications could be of interest in the cosmology of the early universe on time scales well above the Planck time but below characteristic times of nuclear reactions, see~\cite{anderson-molina-david-cook,anderson-molina-sanders} for related work.

A much more ambitious would be the passage to non-cosmological spacetimes.
Space-dependent germs of distributional tensor structures that remain closed under successive applications of the Klein--Gordon operator would have to replace the expansion in radially symmetric homogeneous distributions.
The point-splitting limit of such an approach would result in an infinite hierarchy of coupled PDEs for the germ coefficients.
It is not obvious, if such a system can be set up or even be solved.

Another open question is if it possible to improve upon the techniques employed here to avoid what are essentially assumptions on the spatial analyticity of the state.
Since the difference $\mathcal{G}-\mathcal{H}$ does not satisfy a PDE away from the diagonal, it is currently not clear how that can be accomplished while at the same time still ensuring the correct regularization of the state and without adding arbitrarily high derivatives of the scale factor to system.
Working in the opposite direction, one could attempt to develop a fully analytic theory (also in time).
Such an analytic theory would inevitably involve analytic Hadamard states, it might offer other ideas to solve the SCE in non-cosmological spacetimes, using e.g.\ the construction of analytic Hadamard states developed in~\cite{gerard-wrochna}.

It is often argued that the higher derivatives of the metric appearing in the SCE can lead to runaway solutions which deviate significantly from classical solutions of the Einstein equation.
Although recent numerical work in collaboration with N.~Rothe suggests otherwise \cite{nicolai}, it might still be worthwhile to apply our methods to the order-reduced SCE (cf.\ \cite{flanagan-wald}), which is sometimes suggested as a better behaved approximation of the interaction between quantum matter and classical spacetime geometry.

Potentially, our new formulation of the SCE also has numerical consequences, as it avoids time integration of rapidly oscillating modes of the quantum state and integration in momentum space.
Theorem~\ref{thm:continuous} establishes certain bounds for the change of the solution under modification of the initial conditions.
This can be used to truncate $\mathcal{M}$ at a certain order with a controlled error for the solution of the SCE.
As the space of moments with zero entries after a prescribed order is invariant under the dynamics \eqref{eq:moment-dynamics}, such approximated initial conditions give rise to a finite dimensional ODE which can be numerically integrated in time using standard methods.
However, convergence rates as a function of the time interval need to be carefully examined to judge numerical viability.

\bigskip
\paragraph{Acknowledgments}
This work was supported by the DFG project ``Solutions and Stability of the semiclassical Einstein equation – a phase space approach''.
The authors thank C.~Bohlen, A.~Daletzki, M.~Friesen, T.~Hack, B.~Jacob, N.~Pinamonti, N.~Rothe and A.~Schenkel for useful discussions.
We also thank the anonymous referees for their valuable and constructive comments.
\bigskip

\appendix
\section{Appendix}

\subsection{Homogeneous distributions}
\label{sec:homogeneous}

Roughly following Chap.~3.2 of~\cite{hormander:1}, we define for $z \in \CC$ with $\Re z > -1$ the function on~$\RR$
\begin{equation}\label{eq:homogeneous-defn1}
  k_+^z \defn \begin{cases*}
    k^z & if $k > 0$, \\
    0   & if $k \leq 0$.
  \end{cases*}
\end{equation}
Since this function is locally integrable, it defines a distribution.
It can be extended to $z \in \CC \setminus \{-1,-2,\dotsc\}$ by analytic continuation.
To further extend $k_+^z$ to all $z \in \CC$, we define for $n \in \NN$ and any test function $f \in C^\infty_{\mathrm c}(\RR)$
\begin{equation}\label{eq:homogeneous-defn2}
  \langle k_+^{-n}, f \rangle \defn \frac{1}{(n-1)!} \biggl( -\int_0^\infty \log(k) f^{(n)}(k) \dif k + f^{(n-1)}(0) \sum_{j=1}^{n-1} \frac{1}{j} \biggr).
\end{equation}
Defined in this way, $k_+^z$ satisfies for $z \in \CC$ the homogeneity property $\langle k_+^z, k f \rangle = \langle k_+^{z+1}, f \rangle$.
Note that~\eqref{eq:homogeneous-defn2} is not the only possible extension of~\eqref{eq:homogeneous-defn1} to all of $z \in \CC$ -- different extensions differ by derivatives of the delta distribution at zero.

To calculate the (inverse) Fourier transform of $k_+^{-n}$, note first that
\begin{equation*}
  \langle k_+^{-1}, \e^{\im k x} \rangle = \lim_{\nu \to 0} \int_0^\infty k^{\nu-1} \e^{\im k x} \dif k - \frac{1}{\nu}.
\end{equation*}
Then we compute
\begin{align*}
  \int_0^\infty k^{\nu-1} \e^{\im k x} \dif k
  &= \Gamma(\nu) (-\im x)^{-\nu}
  = \bigl( \nu^{-1} - \gamma + \mathcal{O}(\nu) \bigr) \bigl( 1 - \nu \log(-\im x + 0) + \mathcal{O}(\nu^2) \bigr) \\
  &= \nu^{-1} - \gamma - \log(-\im x + 0) + \mathcal{O}(\nu)
\end{align*}
Therefore, we find
\begin{equation*}
  \langle k_+^{-1}, \e^{\im k x} \rangle = -\gamma - \log(-\im x + 0) = -\gamma - \log\abs{x} + \frac{\im\uppi}{2} \sgn(x),
\end{equation*}
which, by~\eqref{eq:homogeneous-defn2}, immediately implies that, for $n \in \NN$,
\begin{equation}\label{eq:k+-fourier}
  \langle k_+^{-n}, \e^{\im k x} \rangle = \frac{(\im x)^{n-1}}{\Gamma(n)} \mleft( \psi(n) - \log\abs{x} + \frac{\im\uppi}{2} \sgn(x) \mright),
\end{equation}
where we used the relation between the harmonic numbers and the Digamma function
\begin{equation*}
  \psi(n) = \sum_{j=1}^{n-1} \frac{1}{j} - \gamma,
  \quad
  n \in \NN.
\end{equation*}
For completeness, let us also recall the definition of the Digamma function: $\psi(z) \defn \Gamma'(z) / \Gamma(z)$ for $z \in \CC \setminus \{0, -1, -2, \dotsc\}$.

This motivates the definition of the distributions (for $r \geq 0$)
\begin{equation*}
  h_z(r) \defn \frac{\e^{\im z \uppi/2}}{2\uppi^2} \frac{r^{z-2}}{\Gamma(z)} \bigl( \log(r) - \psi(z) \bigr),
\end{equation*}
which extends analytically to $z \in \CC$.
In particular, we have for $n \in \NN_0$
\begin{equation*}
  h_{-n}(r) = \frac{\im^n n!}{2\uppi^2 r^{n+2}}
\end{equation*}
because both $1/\Gamma(z)$ and $\psi(z)/\Gamma(z)$ are entire functions with
\begin{equation*}
  \frac{1}{\Gamma(-n)} = 0
  \quad\text{and}\quad
  \frac{\psi(-n)}{\Gamma(-n)} = (-1)^{n+1} \Gamma(n+1).
\end{equation*}

Taking the imaginary part of~\eqref{eq:k+-fourier}, we find
\begin{equation*}
  h_n(r) = \frac{1}{2\uppi^2r} \Im \langle k_+^{-n}, k \e^{\im k r} \rangle = \frac{1}{2\uppi^2r} \Im \langle k_+^{-n+1}, \sin(k r) \rangle
\end{equation*}
for even $n$.
Moreover, $h_z$ satisfies the important homogeneity property
\begin{equation}\label{eq:homog_h}
  -\upDelta_r h_{z+2}(r) = h_z(r),
\end{equation}
where we recall that $\upDelta_r$ denotes the (three dimensional) radial Laplacian.

\subsection{Weighted sequence spaces}
\label{sec:weighted-spaces}

Let $w = (w_n)$ be a sequence of strictly positive numbers, called the \emph{weights}.
By $\ell^p(w)$, $p \geq 1$, we denote the space of complex sequences $x = (x_n)$ with convergent norm
\begin{equation*}
  \norm{x}_{p,w} \defn \Bigl( \sum_n {}\abs{w_n^{-1} x_n}^p \Bigr)^{1/p}.
\end{equation*}
If $p = \infty$, we denote by $\ell^\infty(w)$ the space of complex sequences with convergent norm
\begin{equation*}
  \norm{x}_{\infty,w} \defn \sup_n w_n^{-1} \abs{x_n}.
\end{equation*}
These are the \emph{weighted $\ell^p$ spaces}.
If $w_n=1$, we omit the weight and denote by $\ell^p$ with norm $\norm{\,\cdot\,}_p$ the ordinary $\ell^p$ spaces.
Note that $\ell^p(w)$ is reflexive for $1 < p < \infty$.

Consider two weight sequences $v,w$.
Then it is easily seen that
\begin{equation*}
  \norm{x}_{p,v} \leq \sup_n \frac{w_n}{v_n} \norm{x}_{p,w}.
\end{equation*}

One of the most important operators on sequence spaces is the left-shift operator~$L$, formally defined by $L (x_0, x_1, x_2, \dotsc) = (x_1, x_2, \dotsc)$.
If $w$ are weights such that the sequence $(w_{n+1}/w_n)$ of ratios of consecutive weights is bounded, the left-shift operator is bounded on $\ell^p(w)$.
Indeed,
\begin{equation}\label{eq:L-norm}
  \norm{Lx}_{p,w} \leq \sup_n \frac{w_{n+1}}{w_n} \norm{x}_{p,w}.
\end{equation}
However, if $(w_{n+1}/w_n)$ is unbounded, also the left-shift operator is unbounded.

More generally, for $m \in \NN$ and two sequences $v,w$ weights, we have
\begin{equation*}
  \norm{L^m x}_{p,v} \leq \sup_n \frac{w_{n+m}}{v_n} \norm{x}_{p,w}.
\end{equation*}
Finally, note that, if $L$ is unbounded on $\ell^p(w)$, its resolvent set is empty and its point spectrum fills the entire complex plane.

\subsection{Some inequalities}

\begin{lemma}\label{lem:binom-ineq1}
  For $0 \leq m \leq n \in \NN_0$ and $p \in (0,1)$,
  \begin{equation*}
    \binom{n}{m} (1-p)^n \leq \binom{\floor{m/p}}{m} (1-p)^{\floor{m/p}},
  \end{equation*}
  where $\floor{\,\cdot\,}$ denotes the floor function.
\end{lemma}
\begin{proof}
  The case $m = 0$ is obvious.
  For $m > 0$, we calculate the ratio of successive terms on the left-hand side:
  \begin{equation*}
    \frac{\binom{n+1}{m} (1-p)^{n+1}}{\binom{n}{m} (1-p)^n} = \frac{n+1}{n+1-m} (1-p).
  \end{equation*}
  Then we observe that
  \begin{equation*}
    \frac{n+1}{n+1-m} (1-p) \geq 1
  \end{equation*}
  if and only if $n+1 \leq m/p$ to find the maximum at $n = \floor{m/p}$.
\end{proof}

Note that this lemma can also be stated in the language of probability theory: $n = \floor{m/p}$ maximizes the probability of getting exactly $m$ failures for a random variable following the binomial distribution with parameters~$n$ (number of trials) and~$p$ (probability of success).

An important inequality, accurately describing the asymptotics of the Gamma function, is Stirling's inequality
\begin{equation*}
  \sqrt{2\uppi}\, x^{x+\frac12} \e^{-x} < \Gamma(x+1) < x^{x+\frac12} \e^{-x+1},
  \quad x > 0.
\end{equation*}
This double inequality can be improved in various ways, e.g., for $x \geq 1$,
\begin{equation}\label{eq:stirling}
  \sqrt{2\uppi}\, x^{x+\frac12} \e^{-x+(12x+1)^{-1}} < \Gamma(x+1) < \sqrt{2\uppi}\, x^{x+\frac12} \e^{-x+(12x)^{-1}},
\end{equation}
which can be obtained from~\cite{shi:gamma}, where also sharper bounds are presented.

Another useful inequality for the Gamma function is Gautschi's inequality.
For $x > 0$ and $s \in (0,1)$, we have
\begin{equation}\label{eq:gautschi}
  x^{1-s} < \frac{\Gamma(x+1)}{\Gamma(x+s)} < (x+1)^{1-s},
\end{equation}
see e.g.\ (5.6.4) of~\cite{nist}, which follows from the strict log-convexity of the Gamma function.

Combining the inequalities above, we find
\begin{proposition}\label{prop:binomial}
  For $m \in \NN$ and $p \in (0,1)$,
  \begin{equation}\label{eq:max-binomial-ineq}
    \max_n \binom{n+m}{m} p^m (1-p)^n \leq \min\biggl\{1,\frac{2}{\sqrt{2\uppi m}} (1-p)^{-\frac12}\biggr\}
  \end{equation}
  and, as $m \to \infty$,
  \begin{equation}\label{eq:max-binomial-asymp}
    \max_n \binom{n+m}{m} p^m (1-p)^n \sim \frac{1}{\sqrt{2\uppi m}} (1-p)^{-\frac12}.
  \end{equation}
\end{proposition}
\begin{proof}
  Denote by $\{x\} = x - \floor{x}$ the fractional part of a real number~$x$.
  We calculate
  \begin{align*}
    \max_n \binom{n+m}{m} p^m (1-p)^n
    &= \binom{\floor{m/p}}{m} p^m (1-p)^{\floor{m/p}-m} \\
    &\leq \binom{m/p}{m} p^m (1-p)^{m/p-m} \biggl(1 + \frac{p}{m (1-p)}\biggr)^{\{m/p\}} \\
    &< \frac{1}{\sqrt{2\uppi m}} (1-p)^{-\frac12} \biggl(1 + \frac{p}{m (1-p)}\biggr)^{\{m/p\}} \\
    &\leq \frac{1}{\sqrt{2\uppi m}} (1-p)^{-\frac12} \biggl(1 + \biggl\{\frac{m}{p}\biggr\} \frac{p}{m(1-p)} \biggr),
  \end{align*}
  where we applied (in this order) Lem.~\ref{lem:binom-ineq1}, Gautschi's inequality~\eqref{eq:gautschi}, Stirling's inequality~\eqref{eq:stirling} and Bernoulli's inequality.
  Finally, an application of the inequality
  \begin{equation*}
    \biggl\{\frac{m}{p}\biggr\} \frac{p}{m(1-p)} = \biggl\{\frac{m (1-p)}{p}\biggr\} \frac{p}{m(1-p)} \leq \min\biggl\{1, \frac{p}{m(1-p)}\biggr\}
  \end{equation*}
  yields~\eqref{eq:max-binomial-ineq}, and from
  \begin{equation*}
    \lim_{m \to \infty} \biggl\{\frac{m}{p}\biggr\} \frac{p}{m(1-p)} = 0
  \end{equation*}
  we obtain~\eqref{eq:max-binomial-asymp}.
\end{proof}

\begin{proposition}\label{prop:ratio-sum-ineq}
  Let $n \in \NN$.
  We have
  \begin{equation*}
    \sum_{m=1}^{\ceil{\frac{n+1}{2}}} \frac{(m-1)!}{(n-m)!} \leq \begin{cases*}
      \frac32         & odd $n$, \\
      \frac{n}{2} + 1 & even $n$.
    \end{cases*}
  \end{equation*}
\end{proposition}
\begin{proof}
  We consider only the odd case, the proof for the even case proceeds analogously.
  The smallest summand in the sum for $n$ is $1/(n-1)!$, which is larger than the second smallest summand $1/n!$ in the sum for $n+2$.
  Therefore we can bound all sums uniformly in $n$ by summing up the smallest summands for each $n$.
  If we proceed like that for $n \geq 5$, we obtain
  \begin{equation*}
    1 + \frac16 + \sum_{n=2}^\infty \frac{1}{(2n)!} = \cosh 1 - \frac13 < \frac32.
  \end{equation*}
  For $n=1$ and $n=3$, the sums yield $1$ and $\frac32$, respectively.
\end{proof}

\subsection{Expansion of Synge's world function}
\label{sec:expansion}

To compute the (truncated) Hadamard parametrix in a given spacetime, it is necessary (among other things) to find Synge's world function $\sigma$.
The world function $\sigma(x,y)$ is defined as half the square signed geodesic distance between the points $x$ and $y$ (if the two points lie in a geodesically convex neighbourhood) and satisfies the relation
\begin{equation}\label{eq:synge}
  2 \sigma = \bigl((\nabla^\mu \otimes \one) \sigma\bigr) \bigl((\nabla_\mu \otimes \one) \sigma\bigl).
\end{equation}
In fact, the relation~\eqref{eq:synge} together with the coincidence limits
\begin{equation}\label{eq:synge-coinc}
  [\sigma] = 0,
  \quad
  [(\nabla_\mu \otimes \one) \sigma] = 0,
  \quad
  [(\nabla_\mu \nabla_\nu \otimes \one) \sigma] = g_{\mu\nu}
\end{equation}
uniquely defines Synge's world function.

An expansion of Synge's world function $\sigma(x,y)$ in terms of the coordinate distance $\delta x$ between the points $x$ and $y$ can be obtained in the following way~\cite{siemssen:phd}:
We make the Ansatz (in the sense of formal power series)
\begin{equation*}
  \sigma(x,y) = \sum_n \frac{1}{n!} \varsigma_{\mu_1 \dotsm \mu_n}(x) \delta x^{\mu_1} \delta x^{\mu_n}.
\end{equation*}
As a consequence of~\eqref{eq:synge} and~\eqref{eq:synge-coinc}, we find the recurrence relation
\begin{align*}
  2(1-n) \varsigma_{\mu_1 \dotsm \mu_n} & = \sum_{j=2}^{n-2} \binom{n}{j}\, g^{\nu \rho} \bigl(\partial_\nu \varsigma_{(\mu_1 \dotsm \mu_j|} - \varsigma_{(\mu_1 \dotsm \mu_j| \nu}\bigr) \bigl(\partial_\rho \varsigma_{|\mu_{j+1} \dotsm \mu_n)} - \varsigma_{|\mu_{j+1} \dotsm \mu_n) \rho}\bigr) \\&\quad - 2n \partial_{(\mu_1} \varsigma_{\mu_2 \dotsm \mu_n)}
\end{align*}
together with `initial' coefficients $\varsigma = 0, \varsigma_\mu = 0, \varsigma_{\mu\nu} = g_{\mu\nu}$.

\subsection{Some technical proofs}
\label{sec:proofs}

\begin{proof}[of Thm.~\ref{thm:quasilinear}]\label{proof:thm:quasilinear}
  For $\vec{a} = (a_0, a_1, \dotsc, a_k) \in C(I; \mathcal{B}_R)$, define
  \begin{subequations}\label{eq:contraction-map}\begin{align}
    \phi[\vec{a}](\tau) &\defn \bigl( a_1(\tau), a_2(\tau), \dotsc, a_k(\tau), f(\tau, \vec{a}, \mathcal{M}[\vec a]) \bigr), \\
    \Phi[\vec a](\tau) &\defn \vec{a}_\init + \int_{\tau_\init}^\tau \phi[\vec{a}](\eta) \dif\eta.
  \end{align}\end{subequations}
  By assumptions (ii) and (iii), we have
  \begin{equation}\label{eq:solution-estimate}
    \norm{\Phi[\vec a](\tau) - \vec{b}_c} \leq \norm{\vec{a}_\init - \vec{b}_c} + \int_{\tau_\init}^\tau \norm{\phi[\vec{a}](\eta)} \dif\eta \leq \frac{R}{2} + (\tau - \tau_\init) (\norm{\vec{b}_c} + R + \mu_f),
  \end{equation}
  which proves that $\Phi$ closes on $C(J; \mathcal{B}_R)$, where $J = [\tau_\init,\tau_\mathrm{stop}]$ for sufficiently small $\tau_\mathrm{stop} \leq \tau_\fin$.

  If $\vec{a}, \tilde{\vec a} \in C(J; \mathcal{B}_R)$, we calculate (using assumptions (iii) and (iv)) for $\tau \in J$
  \begin{align*}
    \norm{\Phi[\vec a](\tau) - \Phi[\tilde{\vec a}](\tau)}
    &\leq \int_{\tau_\init}^\tau \norm{\phi[\vec{a}](\eta)-\phi[\tilde{\vec a}](\eta)} \dif\eta \\
    &\leq \int_{\tau_\init}^\tau \bigl( (1+L_f) \norm{\vec{a}(\eta)-\tilde{\vec a}(\eta)} + L_f {\norm{\mathcal{M}[\vec a](\eta)-\mathcal{M}[\tilde{\vec a}](\eta)}}_{p,v} \bigr) \dif\eta \\
    &\leq (1+L_f) (\tau_\mathrm{stop}-\tau_\init) {\norm{\vec{a}-\tilde{\vec a}}}_{C(J; \RR^{k+1})} \\&\qquad + L_f L_{\mathcal M} \int_{\tau_\init}^\tau\int_{\tau_\init}^{\eta} \norm{\vec{a}(\zeta)-\tilde{\vec a}(\zeta)} \dif\zeta \dif\eta \\
    &\leq \bigl( (1+L_f) (\tau_\mathrm{stop}-\tau_\init) + \tfrac{1}{2} L_f L_{\mathcal M} (\tau_\mathrm{stop}-\tau_\init)^2 \bigr) {\norm{\vec{a}-\tilde{\vec a}}}_{C(J; \RR^{k+1})}.
  \end{align*}
  Therefore
  \begin{equation*}
    {\norm{\Phi[\vec a] - \Phi[\tilde{\vec a}]}}_{C(J; \RR^{k+1})} \leq \bigl( (1+L_f) (\tau_\mathrm{stop}-\tau_\init) + \tfrac{1}{2} L_f L_{\mathcal M} (\tau_\mathrm{stop}-\tau_\init)^2 \bigr) {\norm{\vec{a}-\tilde{\vec a}}}_{C(J; \RR^{k+1})}.
  \end{equation*}
  which shows that $\Phi$ is a contraction map if $\tau_\mathrm{stop}$ is sufficiently small.

  It now follows by the Banach fixed-point theorem that $\Phi$ has a unique fixed point $\vec{a} \in C^1(J; \mathcal{B}_R)$.
  By the form of \eqref{eq:contraction-map}, $\vec{a}$ is such that $\vec{a} = (a, a', \dotsc, a^{(k)})$ for some $a \in C^{k+1}(J)$.
  This yields the unique solution to~\eqref{eq:quasilinear} with the properties stated in the theorem.
\end{proof}

\begin{proof}[of Thm.~\ref{thm:continuous}]\label{proof:thm:continuous}
  The two solutions $a(\tau)$ and $\tilde{a}(\tau)$ satisfy
  \begin{align*}
    \vec{a}(\tau) &= \vec{a}_\init + \int_{\tau_\init}^\tau \bigl( a'(\eta), \dotsc, a^{(k)}(\eta), f(\eta, \vec{a}, \mathcal M[\vec a]) \bigr) \dif\eta, \\
    \tilde{\vec a}(\tau) &= \tilde{\vec a}_\init + \int_{\tau_\init}^\tau \bigl( \tilde{a}'(\eta), \dotsc, \tilde{a}^{(k)}(\eta), \tilde{f}(\eta, \tilde{\vec a}, \tilde{\mathcal M}[\tilde{\vec a}]) \bigr) \dif\eta
  \end{align*}
  and thus
  \begin{equation*}
    \norm{\vec{a}(\tau) - \tilde{\vec a}(\tau)} \leq \norm{\vec{a}_\init - \tilde{\vec a}_\init} + \int_{\tau_\init}^\tau \bigl( \norm{\vec{a}(\eta) - \tilde{\vec a}(\eta)} + \abs{f(\eta, \vec{a}, \mathcal{M}[\vec a]) - \tilde{f}(\eta, \tilde{\vec a}, \tilde{\mathcal M}[\tilde{\vec a}])} \bigr) \dif\eta.
  \end{equation*}
  Since
  \begin{align*}\MoveEqLeft
    \abs{f(\tau, \vec{a}, \mathcal{M}[\vec a]) - \tilde{f}(\tau, \tilde{\vec a}, \tilde{\mathcal M}[\tilde{\vec a}])} \\
    &\leq \abs{f(\tau, \vec{a}, \mathcal{M}[\vec a]) - \tilde{f}(\tau, \vec{a}, \mathcal{M}[\vec a])} + \abs{\tilde{f}(\tau, \vec{a}, \mathcal{M}[\vec a]) - \tilde{f}(\tau, \tilde{\vec a}, \tilde{\mathcal M}[\tilde{\vec a}])} \\
    &\leq \mu_f' + L_f \bigl( \norm{\vec{a}(\tau) - \tilde{\vec a}(\tau)} + \norm{\mathcal{M}[\vec a](\tau) - \tilde{\mathcal M}[\tilde{\vec a}](\tau)}_{p,v} \bigr) \\
    &\leq \mu_f' + L_f \norm{\vec{a}(\tau) - \tilde{\vec a}(\tau)} + L_f L_\mathcal{M}' \mleft( \norm{\mathcal{M}_\init - \tilde{\mathcal M}_\init}_{p,w} + \int_{\tau_\init}^\tau \norm{\vec{a}(\eta) - \tilde{\vec a}(\eta)} \dif\eta \mright)
  \end{align*}
  with the constants as defined in the theorem, we find
  \begin{align*}\MoveEqLeft
    \norm{\vec{a}(\tau) - \tilde{\vec a}(\tau)} - \norm{\vec{a}_\init - \tilde{\vec a}_\init}
    \\&\leq \int_{\tau_\init}^\tau \biggl( \mu_f' + L_f L_\mathcal{M}' \norm{\mathcal{M}_\init - \tilde{\mathcal M}_\init}_{p,w} + \bigl((1 + L_f) + L_f L_\mathcal{M}' (\tau - \eta) \bigr) \norm{\vec{a}(\eta) - \tilde{\vec a}(\eta)} \biggr) \dif\eta.
  \end{align*}
  This implies inequality~\eqref{eq:dependence-ineq} by a generalization of Gr{\"o}nwall's inequality (e.g., Thm.~1.4 of~\cite{bainov-simeonov}).
\end{proof}


\end{document}